\documentclass[twoside,leqno,twocolumn]{article}

\usepackage[letterpaper]{geometry}

\usepackage{ltexpprt}

\usepackage{booktabs} 
    
\newlength{\bracewidth}

\usepackage[ruled]{algorithm}
\usepackage[noend]{algpseudocode}
\usepackage{amsmath}
\usepackage{amssymb}
\usepackage{physics}
\usepackage{color}
\usepackage{graphicx}
\usepackage{subfig}
\usepackage{epsfig}
\usepackage{epstopdf}
\usepackage{amsfonts}
\usepackage{latexsym}
\usepackage{amssymb}
\usepackage{amsmath}
\usepackage{multirow}
\usepackage{mathtools}
\usepackage[export]{adjustbox}
\usepackage[T1]{fontenc}
\usepackage{hyperref}
\usepackage{framed}

\usepackage{comment}
\usepackage{fullpage}
\usepackage[table]{xcolor}

\newtheorem{observation}{Observation}[section]

\newcommand{\eps}{\epsilon}

\newcommand{\cone}{\mathcal{C}}
\newcommand{\parset}{\mathcal{P}}

\title{Search and evacuation with a near majority of faulty agents\thanks{This is the full version of the paper which appeared in~\cite{inbookryan}} \thanks{J. Czyzowicz, R. Killick, and E. Kranakis were supported in part by NSERC grants.}
}

\author{J. Czyzowicz\thanks{Dep. d' Informatique, University\'{e} du Quebec en Outaouais, QC, Canada} 
\and R. Killick\thanks{School of Computer Science, Carleton University, Ottawa, ON, Canada} 
\and E. Kranakis\thanks{School of Computer Science, Carleton University, Ottawa, ON, Canada}  
\and G. Stachowiak\thanks{Institute of Computer Science, University of Wroclaw, Wroclaw, Poland}}
\date{}

\begin{document}
\maketitle

\begin{abstract}
        There are $n\geq 3$ unit speed mobile agents placed at the origin of the infinite line. In as little time as possible, the agents must find and evacuate from an exit placed at an initially unknown location on the line. The agents can communicate in the wireless mode in order to facilitate the evacuation (i.e. by announcing the target's location when it is found). However, among the agents are a subset of at most $f$ crash faulty agents who may fail to announce the target when they visit its location.
        
        In this paper we study this aforementioned problem for the specific case that $n=2f+1$. We introduce a novel type of search algorithm and analyze its competitive ratio -- the supremum, over all possible target locations, of the ratio of the time the agents take to evacuate divided by the initial distance between the agents and the target. In particular, we demonstrate that the competitive ratio of evacuation is at most $7.437011$ for $(n,f)=(3,1)$; at most $7.253767$ for $(n,f)=(5,2)$ and $(7,3)$; and at most $7.147026$ for $(n,f)=(9,4)$. For larger values of $n=2f+1$ we prove an asymptotic upper bound of $4+2\sqrt{2}$. We also adapt our evacuation algorithm for $(n,f)=(3,1)$ to the problem of search by three agents with one byzantine fault, i.e. the faulty agent may also lie about finding the target. In doing so we improve the best known upper bound on this search problem from 8.653055 to 7.437011.
\end{abstract}

\section{Introduction}\label{sec:intro}
        Problems of search and exploration are central to many areas of computer science and mathematics and, accordingly, have received much attention in the literature. Perhaps the simplest search type problem considers the optimal trajectory of a single mobile agent tasked with finding a target placed at an unknown location on the infinite line. The goal of the agent is to minimize the competitive ratio -- the supremum over all possible target locations of the ratio of the time the agent takes to find the target and the initial distance between the agent and target. Independently studied by Bellman and Beck in the 1960's, it is now well known that the optimal trajectory for this single agent search uses a doubling strategy whereby the agent, starting at the origin, moves between points on the line at alternating positions $1,-2,4,-8,\ldots$. It is a simple task to show that this trajectory ensures a competitive ratio of $9$.
        
        Search by multiple agents on the line is a natural extension of the single agent search. Of course, with more than one agent also arise questions about how search is affected by the presence of agents with differing capabilities/attributes. For example, one can consider agents with different speeds, and or communication abilities. A particularly interesting and important topic in group search is the development of fault tolerant search algorithms.

        In this paper we study a version of fault tolerant group search on the line. Specifically, we consider the problem of {\em evacuation} by $n=2f+1$ mobile agents when at most $f$ of these agents are faulty. The agents all begin the search at the same time from a common location and the goal is for the agents to find and exit from a target placed at an unknown location on the line. To achieve this goal the agents can co-operate by exchanging messages with one another in wireless mode (i.e. instantaneously and across any distance). This goal is impeded by the presence of $f$ crash faulty agents who may fail to announce that they have found the target when they detect it.

        We also use the results from the crash evacuation problem to improve upper bounds on the problem of search by three agents at most one of which is {\em byzantine} faulty. A byzantine faulty agent is similar to a crash faulty agent except that byzantine agents can also lie about finding the target. When $n=2f+1$ it can happen that all agents are required to reach the target in order for the search to complete and so crash evacuation can be viewed as a sub-problem of the more difficult byzantine search problem.

\subsection{Model}\label{sec:model}
        We have $n=2f+1$ mobile agents with at most $f$ of them faulty. Robots/agents all begin at a common location referred to as the origin. The agents can move up to a maximum unit speed in either the positive direction (referred to as moving to the right) or the negative direction (referred to as moving to the left) and an agent may change its travel direction arbitrarily often. There is no time cost associated with an agent changing its movement direction. 

        The agents are labelled with unique identifiers taken from the set $\{0,\ldots,n-1\}$ and can communicate with each other in the wireless mode. A parallel search algorithm specifies a unique trajectory for each agent and all agents are assumed to have full knowledge of these trajectories. Since agents know of all other trajectories it follows that the only kind of message broadcast by an agent will be a notification that it has detected the target at its current location. If an agent does not broadcast a message while visiting a location then it is assumed that the agent did not detect the target at that location.

        Each agent is aware of the number $f$ of faults, however, the identity of the faulty agents is unknown. The fault model considered for the agents is that of crash or silent faults. In this model an agent may fail to announce the target when it is detected, however it cannot send a message falsely claiming that it has found the target when it has not (this is known as the byzantine model). The presence of faulty agents thus implies that an agent cannot necessarily trust that the target is not at a location previously visited by another agent. In order to be sure that a target is not at a particular location $x$, it will be required that the location $x$ has been visited by at least one provably reliable agent. With at most $f$ faults, at least $f+1$ agents must visit $x$ in order to have this guarantee.

        It is possible that the faulty agents do not follow the trajectories assigned to them. However, as all agents are aware of the trajectories of the other agents, any agent that is found to be not following its assigned trajectory can be reliably identified as faulty. We will assume that the identity and behaviour of the faulty agents is controlled by an adversary who will always act in a way to maximize the competitive ratio. We may therefore safely assume that any such premature identification will not occur. Each agent will therefore follow the trajectory initially assigned to it until they either find the target or they receive an announcement that the target has been found. Since announcements can always be trusted, agents will immediately move to the announced location in order to complete the evacuation.

\subsection{Preliminaries and notation}\label{sec:prelim}
        We begin with some definitions.
        \begin{Definition}\label{def:evac}
                The evacuation time $E^x_{f}$ of a parallel search algorithm for $n=2f+1$ agents, at most $f$ of which are faulty, is the worst case time required until the last reliable agent reaches a target at location $x$.
        \end{Definition}
        \begin{Definition}\label{def:search}
                The search time $S^x_{f}$ of a parallel search algorithm for $n=2f+1$ agents, at most $f$ of which are faulty, is the worst case time required for the first reliable agent to reach $x$.
        \end{Definition}
        \begin{Definition}\label{def:cr}
                The competitive ratio is defined as $R_{f} = \sup_{x} \frac{E^x_{f}}{|x|}$ and represents the worst case ratio of the evacuation time to the lower bound $|x|$ on the time required to find the target.
        \end{Definition}

        Since crash faulty agents fail silently (i.e. they cannot lie about finding the target), it follows that any announcement made by an agent must be truthful. As a result the only sensible thing for the (reliable) agents to do once an announcement has been made is to immediately move to the announced target's location. This observation has two important implications. First, it implies that we can define a parallel evacuation algorithm entirely by the trajectories of the agents. Second, it implies that we can express the evacuation time as the sum of $S_f^x$ and the distance between the target and the agent most distant from the target at the time $S_f^x$. This last point leads us to make the following definition.
        \begin{Definition}\label{def:delta}
                Given a parallel search algorithm for $n=2f+1$ agents, at most $f$ of which are faulty, define $i^x_f$ and $\Delta^x_f$ as the identity of, and distance between, the agent most distant from location $x$ at the time $S_f^x$.
        \end{Definition}
        With this definition we can express the evacuation time as follows
        \begin{equation}\label{eq:evac_gen}
                E = S_f^x + \Delta_f^x.
        \end{equation}

        We consider agent trajectories defined by sets of turning points -- points on the line at which agents change their movement direction, and between which the agents move at constant unit speed. We use the notation $d_{i,j}$ to refer to the turning point $j$ of agent $i$. We will assume that the turning points alternate on either side of the origin with increasing absolute values, i.e. if $d_{i,j} > 0$ then $d_{i,j+1}<0$ and $|d_{i,j+1}| > |d_{i,j}|$. 

        For these types of trajectories one must make additional assumptions in order to achieve a constant competitive ratio. To see why this is, imagine we have a set of trajectories with first turning points $d_{i,0}$ and assume that for the majority of the agents we have $d_{i,0} > \delta$ for some $\delta>0$. Then the target can be placed at location $-\eps$ with $\eps>0$ arbitrarily small and all agents that initially moved to the left are made to be faulty. The first time a reliable agent can reach the target is then $2\delta$ and the competitive ratio is at least $2\delta/\eps$. 
        
        To overcome this problem one usually makes the assumption that the agents are aware of a lower bound on the distance to the target. Then, by making the first turning points much smaller than this lower bound, a finite competitive ratio is possible. Alternatively, one can assume that the agents do not have first turning points. In other words, one assumes that the turning point sequence $d_{i,j}$ extends to $j=-\infty$ and the agents have always been moving. Although less realistic, we find the latter assumption to be more elegant mathematically and we will take this approach here. 
                
        We end this section with a lemma which specifies how the target will be placed in the worst case.
        \begin{lemma}\label{lm:wc}
                The supremum of $E_f^x/|x|$ always occurs when $x$ is a turning point.
        \end{lemma}

\subsection{Related work}\label{sec:related_work}
        Search problems are optimization problems generally concerned with minimizing the time required for a set of mobile agents to find a hidden target in a given environment. One usually assumes that the environment is known in advance and the focus is on studying the effects on the search time under different assumptions on the agent capabilities. Searching in an unknown environment implies exploration where quite often there are additional/alternative goals the agents are required to achieve, e.g. mapping and/or positioning the searchers within the environment \cite{AH00,albers2002exploring,DKP91,HIKK01}.

        When the environment is known, search is a pure optimization problem. The study of search by a single agent on the infinite line was initiated independently by Bellman \cite{bellman1963optimal} and Beck \cite{beck1964linear,beck1965more,beck1970yet} where, among other things, the authors demonstrate the now well known result that a single searcher cannot find a hidden target at initial distance $d$ from the searcher in time less than $9d$. The work by Bellman and Beck gave rise to a number of variants of search on the line. Notable is the work of Heath and Fristedt \cite{fristedt_heath_1974}, Fristedt \cite{fristedt1977hide}, and Gal \cite{gal1972general}. Also notable is the works of Baeza-Yates et. al. \cite{baezayates1993searching,baezayates1995parallel} where, among other things, the authors study problems of search by agents in environments different from the line, e.g. in the plane or at the origin of $w$ concurrent rays (known as the ``Lost Cow'' problem). Group search was initiated in \cite{chrobak2015group} where the problem of evacuation by multiple agents that can communicate face-to-face was studied. More recently, search on the line was considered when: the agents have distinct speeds \cite{bampas2019linear}; turning costs are included \cite{demaine2006online}; the concern is to minimize the energy consumed during the search \cite{CGKKKLNOS19wireless,czyzowicz2021time}.

        Search on the line with possibly faulty searchers was initiated in \cite{PODC16} wherein the authors introduce optimal trajectories -- the {\em proportional schedules} -- for search by $n$ agents at most $f$ of which are crash faulty. This work is particularly relevant to the problem we study here. It should also be noted that the optimality of the proportional schedules for search was only established at a later time in \cite{kupavskii2018lower}. Search with byzantine faults was first studied in \cite{isaacCzyzowiczGKKNOS16} wherein the authors prove a number of lower bounds and upper bounds on the problem. Many of these upper bounds were later improved in \cite{Sun2020}, where, in particular, the authors demonstrate that the proportional schedules of \cite{PODC16} can be used to achieve an upper bound of 8.653055 on the problem of search by three agents, one of which is byzantine faulty.

\subsection{Results and outline}\label{sec:outline}

        Our main result is the development and analysis of a novel search type algorithm for the evacuation problem with crash faulty agents. A summary of the resulting upper bounds are listed in Table~\ref{tbl:bounds} along with the best known lower bounds.
        \begin{table}[tbhp]
                \centering
                \caption{Best known upper and lower bounds on the competitive ratio of evacuation by $n=2f+1$ agents. Upper bounds are due to this work. Lower bounds are due to \cite{PODC16} and \cite{kupavskii2018lower}. The last column displays the reference to the theorem from which the upper bounds derive from. \label{tbl:bounds}}
                \begin{tabular}{c|c|c|c|c}
                        $n$ & $f$ & UB & LB & Theorem\\\hline
                        3 & 1 & 7.437011 & 5.233069 \& \ref{thm:31_crash}\\
                        5 & 2 & 7.253767 & 4.434326 & \multirow{3}{*}{\ref{thm:f_crash}}\\
                        7 & 3 & 7.253767 & 4.076343 & \\
                        9 & 4 & 7.147026 & 3.870110 & 
                \end{tabular}
        \end{table}
        We also prove an asymptotic upper bound on the evacuation by $n=2f+1$ agents of $4+2\sqrt{2}$ (Theorem~\ref{thm:asymp}) and improve the upper bound on search by three agents, at most one of which is byzantine faulty, from $8.653055$ to $7.437011$ (Theorem~\ref{thm:31_search}). The best known lower bound on this search problem is $5.233069$.

        In Section~\ref{sec:prop_schedules} we analyze the competitive ratio of evacuation for the proportional schedules -- a family of trajectories first developed in \cite{PODC16} for the purpose of search by crash faulty agents. We use this section to prove our asymptotic upper bound (Theorem~\ref{thm:asymp}) and also to build intuition on how we can improve upon these trajectories. In Section~\ref{sec:gen_prop_schedules} we introduce a generalization of the proportional schedules and analyze separately the cases that $f > 1$ (Subsection~\ref{sec:many_faults}) and $f=1$ (Subsection~\ref{sec:one_fault}). In Section~\ref{sec:search} we show that our evacuation algorithm for three agents also leads to an improvement on the competitive ratio of search by three agents, one of which is byzantine faulty. Finally, in Section~\ref{sec:conclusions} we conclude with a brief discussion of open problems. Due to space considerations, many figures and proofs of lemmas/theorems have been moved to the appendix.

\section{Crash faulty evacuation}\label{sec:crash_evac}

With the evacuation time expressed as in \eqref{eq:evac_gen} it is clear that in order to optimize the evacuation time one needs to consider a trade-off between the search time $S_f^x$ and the distance $\Delta_f^x$. This is in contrast to the normal search problem which aims only to optimize $S_f^x$. Nevertheless, one can imagine that an algorithm that optimizes $S_f^x$ would still provide a good starting point for studying the evacuation problem. Since it just so happens that an optimal algorithm for crash faulty search is known, we will use this approach to study the evacuation problem.

\subsection{Proportional schedules}\label{sec:prop_schedules}
        An optimal algorithm for the crash faulty search problem was introduced\footnote{The algorithm was later proven to be optimal in \cite{kupavskii2018lower}.} in \cite{PODC16}. This algorithm is referred to as a {\em proportional schedule} and is defined by the collection of $n$ trajectories represented by the sequences of turning points
        \begin{equation}\label{eq:dij}
                d_{i,j} = r^{2i/n} (-r)^j
        \end{equation}
        where $r>1$ is a real number parameter. Figure~\ref{fig:example_prop} depicts example proportional schedules for $n=5,7,9$ 
        using a space-time diagram which plots an agent's position on the $x$ axis with time on the $y$-axis.

        Our strategy for analyzing this algorithm derives from equation \eqref{eq:evac_gen}. We will first compute $S_f^x = S_f^x(r)$ and $\Delta_f^x = \Delta_f^x(r)$. Our goal will be to prove the following Theorem~\ref{thm:crash_evac_prop}. Note that we will not analyze the exact evacuation time at this point since as we will eventually describe a better algorithm.
        \begin{theorem}\label{thm:crash_evac_prop}
                For all $\eps > 0$ the evacuation time of the proportional schedules satisfies
                \[R_f \geq 1+\frac{2r}{r+r^{2/n}} + \frac{4r^{2+1/n}}{(1+\eps)(r+r^{2/n})(r-1)},\]
                \[R_f < 1+\frac{2r}{r+1} + \frac{4r^{2+1/n}}{(r+1)(r-1)}.\]          
        \end{theorem}

        We make note of the following properties of the turning-points $d_{i,j}$ which hold for any agent $i$, turning point $j$, and integer $k$
        \begin{equation}\label{eq:dij_p1}
                d_{i+k,j} = r^{2k/n} d_{i,j},\qquad d_{i,j+k} = (-r)^k d_{i,j}.
        \end{equation}
        Particularly useful is the fact that
        \begin{equation}\label{eq:dmod}
                d_{i+k n,j} = r^{2k} d_{i,j} = d_{i,j+2k}
        \end{equation}
        which allows us to refer to a turning-point $j$ of an agent with ``label'' $i \geq n$ or $i < 0$ with the understanding that we are actually referring to a later/earlier turning-point of agent $i \mod n$. We will make use of these last three properties often and without reference.

        \begin{lemma}\label{lm:tij}
                Define $t_{i,j}$ to be the first time at which agent $i$ reaches location $d_{i,j}$. Then we have
                \[t_{i,j} = \frac{r+1}{r-1}|d_{i,j}|.\]
        \end{lemma}
        We observed in the caption of Figure~\ref{fig:example_prop} that the turning points of the agents all lie along a common cone. This is evidenced by the ratio $t_{i,j}/|d_{i,j}| = \frac{r+1}{r-1}$ being independent of $i$ and $j$. We will make use of this property shortly.

        We define the interval $I_{i,j}$ as follows.
        \begin{Definition}\label{def:Iij}
                The interval $I_{i,j}$ is defined as the semi-open interval
                \[I_{i,j} := (d_{i,j}, d_{i+1,j}] = (1,r^{2/n}]\cdot d_{i,j}.\]
        \end{Definition}
        The sequence of intervals $[I_{i,j}]_{i=-\infty}^{\infty}$ with $j$ even (resp. $j$ odd) covers the entire positive (resp. negative) half-line without overlap. Thus, for any fixed $j$ and position $x$ there exists a unique integer $i$ for which the position $x \in I_{i,j}$. Using the property \eqref{eq:dmod}, we can equivalently say that there exists a unique $j$ and $i \in \{0,\ldots,n-1\}$ for which $x \in I_{i,j}$. Note also that the symmetry inherent to the trajectories implies that for a fixed $z \in (1,r^{2/n}]$ the competitive ratio when the target is placed at $x = z d_{i,j}  \in I_{i,j}$ will not depend on $i$ or $j$.

        The next lemma will allow us to compute the search time of the algorithm.
        \begin{lemma}\label{lm:Tijk}
                Define the time $T_{i,j,k}(z)$ as the time of the first visit by agent $i+k$, $k = 1,\ldots,n$, to the location $z d_{i,j} \in I_{i,j}$. Then
                \begin{equation}\label{eq:Tijk}
                        T_{i,j,k}(z) = \left(z + \frac{2r^{2k/n}}{r-1}\right)|d_{i,j}|.
                \end{equation}
        \end{lemma}

        We now turn our attention to the distance $\Delta_f^x$. We will only bound this distance and to do this we inspect the space-time points at which the trajectories of the agents intersect. 
        \begin{lemma}\label{lm:rho_tau}
                Respectively define $\rho_{i,j,k}$ and $\tau_{i,j,k}$, $k=0,1\ldots,f$, as the position and time at which the trajectory of agent $i$ intersects the trajectory of agent $i+k$ while agent $i$ is traveling away from its turning point $d_{i,j}$ and agent $i+k$ is traveling towards its turning point $d_{i+k,j}$. Then
                \[\rho_{i,j,k} = \frac{r-r^{2k/n}}{r-1}d_{i,j},\quad \tau_{i,j,k} = \frac{r+r^{2k/n}}{r-1}|d_{i,j}|.\]
        \end{lemma}
        The space-time points $(\rho_{i,j,k},\tau_{i,j,k})$ for various $k$ are depicted in Figure~\ref{fig:intersects}. One can observe that for a fixed $k$ the points $(\rho_{i,j,k},\tau_{i,j,k})$ all lie along a common cone. Let $\cone_k$ represent the cone corresponding to points $(\rho_{i,j,k},\tau_{i,j,k})$. Then the slope of $\cone_k$ is $\beta_k := \tau_{i,j,k}/|\rho_{i,j,k}| = \frac{r+r^{2k/n}}{r-r^{2k/n}}$. By referring to Figure~\ref{fig:intersects} one can observe that at all times $t>0$ and for each $k=0,\ldots,f-1$, there exists an agent on either side of the origin in the annular region bounded by cones $\cone_k$ and $\cone_{k+1}$ (this fact also follows easily from the definition of $(\rho_{i,j,k},\tau_{i,j,k})$). Of particular interest is the fact that there always exists an agent located between the cones $\cone_0$ and $\cone_1$ since this agent is the most distant from the origin (on its respective side).
        \begin{observation}\label{obs:cone_bounds}
                At all times $t > 0$ there exists an agent located within each of the intervals $\pm\left[\frac{t}{\beta_1},\ \frac{t}{\beta_0}\right]$.
        \end{observation}
        This observation then easily leads to the following bound on $\Delta^x_f$.
        \begin{lemma}\label{lm:delta}
                \[|x| + \frac{S^x_f}{\beta_1} \leq \Delta_f^x \leq |x| + \frac{S^x_f}{\beta_0}.\]
        \end{lemma}
        We are now in a position to prove Theorem~\ref{thm:crash_evac_prop}. The proof can be found in the appendix.

        To prove Theorem~\ref{thm:crash_evac_prop} we used the fact that the agent most distant from the target at time $S^x_f$ will be somewhere between the cones $\cone_0$ and $\cone_1$. Since $\lim_{n \rightarrow \infty} r^{2/n} = 1$ it is clear that $\lim_{n\rightarrow \infty} \beta_1 = \beta_0$, i.e. the cones $\cone_0$ and $\cone_1$ approach each other as $n$ gets large. This implies that the bounds of Theorem~\ref{thm:crash_evac_prop} will also approach each other for large $n$. This immediately leads to the following conclusion.
        \begin{theorem}\label{thm:asymp}
                The asymptotic competitive ratio of the proportional schedule algorithm is
                \[\hat{R} := \lim_{f \rightarrow \infty} R_f = 7 - \frac{2(r-3)}{r^2-1}.\]
                In particular, if we take $r = 3+2\sqrt{2}$ then $\hat{R} = 4+2\sqrt{2}$.
        \end{theorem}

        To compute an exact expression for the competitive ratio as a function of $r$ we would need to determine the identity $i_f^x$ of the agent that is most distant from $x$ at the time $S_x^f$. Although this is not very difficult to do, it does not add to the results of the paper since we will be improving upon this algorithm in the next section. It is useful, however, to know the optimal competitive ratio for this algorithm for the sake of comparison and discussion. Figure~\ref{fig:cr_compare} in the appendix shows a plot of the optimized competitive ratios as a function of $f$ along with the bounds from Theorem~\ref{thm:crash_evac_prop} evaluated with the optimized parameter $r$. Also shown is the competitive ratio when $r$ is chosen to optimize the search time only. One can observe that in all cases except $f=1$, the actual competitive ratio is essentially identical to the lower bound implying that the optimal choice of $r$ places the agent $i_f^x$ on or near the interior cone $\cone_1$ at the time $S_f^x$. This observation leads one to question whether or not the single degree of freedom provided by the parameter $r$ is sufficient to facilitate an efficient trade-off between the search time and the distance $\Delta_f^x$. In the next section we validate this concern and show that a generalized form of the proportional schedule leads to an improvement in the competitive ratio.

\section{Generalized (proportional) schedules}\label{sec:gen_prop_schedules}
        In this section we consider a generalized form of the proportional schedule. Put simply, we will add an extra two turning points between each pair of turning points $d_{i,j}$ and $d_{i,j+1}$ of the normal proportional schedule. We will refer to these ``sub-turning points'' using the notation $d^{(\ell)}_{i,j}$, $\ell=0,1,2$. An intuitive parameterization of the turning points uses parameters $s \in [0,r+1]$ and $a \in [s-1,r]$ as follows
        \begin{equation}\label{eq:dij_gen}
                d^{(\ell)}_{i,j} = d_{i,j} \cdot \begin{cases}
                        1,& \ell=0\\
                        -a,& \ell=1\\
                        s-a,& \ell=2
                \end{cases},
        \end{equation}
        The parameter $s$ controls the distance between $d^{(1)}_{i,j}$ and $d^{(2)}_{i,j}$, and $a$ controls the location of $d^{(1)}_{i,j}$ (relative to $d_{i,j}$). With the bounds given on $s$ and $a$ we will have $d^{(1)}_{i,j} \in [d_{i,j},d_{i,j+1}]$ and $d^{(2)}_{i,j} \in [d_{i,j},d_{i,j+1}]$. One can also observe that when $s=0$ these trajectories are identical to those of the proportional schedule and so this can be rightly called a generalization.  
        
        Although the parameterization using $(r,s,a)$ is intuitive, it will be much more convenient to replace $s$ with the parameter $q := \frac{r+s}{r-1}$. We will use both of these parameterizations, however, we will favor the one with $q$. For the parameterization with $q$ we have
        \begin{equation}\label{eq:dij_gen}
                d^{(\ell)}_{i,j} = d_{i,j} \cdot \begin{cases}
                        1,& \ell=0\\
                        -a,& \ell=1\\
                        q(r-1)-r-a,& \ell=2
                \end{cases},
        \end{equation}
        We make note of the following identities concerning the parameters $q$ and $s$.
        \begin{equation}
                q = \frac{r+s}{r-1},\quad q-1 = \frac{1+s}{r-1},\quad q+s = r(q-1).
        \end{equation}
        We will use these identities repeatedly and without reference
        
        One approach to analyzing these trajectories would be to compute and optimize the competitive ratio as a function of the parameters $(r,q,a)$. This would involve a nightmarish case analysis that would scare away even the most interested readers. Alas, this is not the approach we take. Instead we will describe sets of objectively good choices for the parameters $q$ and $a$ and express the competitive ratio of the resulting trajectories as a function of the parameter $r$. We will use the results from the previous section to guide us as much as possible. The analysis of the cases $f=1$ and $f>1$ is sufficiently different to warrant considering each separately. We will begin with the case that $f>1$ since this case closely mirrors that of the vanilla proportional schedules.

\subsection{Many faults}\label{sec:many_faults}
        Recall that the worst case scenario for the proportional schedules occurs when the target is placed just beyond a turning point $d_{i,j}$ and the first $f$ agents that visit the target are faulty. We will refer to this scenario as Scenario A in order to refer to it quickly. Also recall that, in all cases that $f>1$, it was optimal to choose $r$ so that the agent $i_f^x$ was located on or near the inner bounding cone $\cone_1$ at the time $S_f^x$ in order to minimize the distance $\Delta_f^x$ in the event that Scenario A occurs. Scenario A will still be a potential worst case for the generalized schedule, however, we can now use our extra degrees of freedom to ensure that agent $i_f^x$ is located on the cone $\cone_1$ at the time $S_f^x$ while leaving the parameter $r$ to facilitate a more efficient tradeoff between $\Delta_f^x$ and $S_f^x$. Our goal is to prove the following theorem.
        \begin{theorem}\label{thm:f_crash}
                Fix the number of faults $f>1$ with $n=2f+1$. Define the functions
                \begin{equation}\label{eq:hat_s}
                        \hat{q}(r,u) := \frac{r^{2u/n-1}+1}{(1+r^{2/n})r^{2(u-1)/n-1}-2r^{1/n}}
                \end{equation}
                \begin{equation}\label{eq:hat_a}
                        \hat{a}(r,q) := \begin{cases}
                        q(r^{1-2/n}-1),& q \leq \frac{r}{r-r^{2/n}}\\
                        q(r^{1-4/n}-1),& \mbox{otherwise}
                \end{cases}
                \end{equation}  
                and the set
                \begin{multline}\label{eq:P}
                        \parset := \Biggl\{(r,u)\ |\ r>1,\ u \in \{f+3,\ldots,n\},\\ \frac{r}{r-1} \leq \hat{q}(r,u) \leq \min\left\{\frac{r}{r-r^{1-2/n}},\frac{r}{r-r^{4/n}}\right\}\Biggr\}
                \end{multline}
                Then, for pairs $(r,u) \in \parset$, the competitive ratio of the generalized proportional schedule with parameters $q = \hat{q}(r,u)$, and $a = \hat{a}(r,\hat{q}(r,u))$ is
                \begin{equation}
                        R_f = \begin{cases}
                                R^A_f,& q \leq \frac{r}{r-r^{2/n}}\\
                                \max\{R^A_f,\ R^B_f\},&\mbox{otherwise}.
                        \end{cases}
                \end{equation}
                where
                \[R^A_f = 1 + \frac{2q(1 + 2qr^{1/n})}{q + (q-1)r^{2/n}}, \quad\mbox{ and }\]
                \[R^B_f \leq 3 + \frac{2(q-1)}{q(r^{1-4/n}-1)}\left[2r - \frac{r^{3/n}(1 + 2qr^{1/n})}{q + (q-1)r^{2/n}}\right].\]
        \end{theorem}

        The rough idea behind this theorem is as follows. Taking $q=\hat{q}(r,u)$ ensures that the agent $i+u$ will be located on the cone $\cone_1$ (properly modified for the generalized schedules) in the event that Scenario A occurs (see Lemma~\ref{lm:s_choice}). Choosing $a = \hat{a}(r,q)$ ensures that agent $i+u$ will be the most distant agent in the event of Scenario A (see Lemma~\ref{lm:intersects_cond}). The quantity $R^A_f$ gives the competitive ratio of Scenario A and $R^B_f$ gives the competitive ratio of an additional potential worst case. One can refer to Figure~\ref{fig:traj} at the end of this subsection for an illustration of the optimized trajectories resulting from these parameter choices for the cases $f=2,3,4$.

        To proceed we need to define analogues of the quantities $t_{i,j}$, $T_{i,j,k}(z)$, $\tau_{i,j,k}$, and $\rho_{i,j,k}$ in the context of the generalized schedules. 

        \begin{lemma}\label{lm:tijell}
                The time $t^{(\ell)}_{i,j}$ at which agent $i$ reaches its sub-turning point $d^{(\ell)}_{i,j}$ is
                \begin{equation}\label{eq:tij_ell}
                        t^{(\ell)}_{i,j} = |d_{i,j}|\cdot\begin{cases}
                                2q-1,& \ell=0\\
                                2q+a,& \ell=1\\
                                q(r+1)-r+a,& \ell=2 
                        \end{cases}.
                \end{equation}       
        \end{lemma}

        \begin{lemma}\label{lm:tijk_gen}
                The time $T_{i,j,k}(z)$ at which agent $i+k$ first reaches location $z = zd_{i,j} \in I_{i,j}$ is
                \begin{equation}\label{eq:Tijk_gen}
                        T_{i,j,k}(z) = \begin{cases}
                        T^\circ_{i,j,k}(z),& a \geq \frac{zr}{r^{2k/n}}\\
                        T^+_{i,j,k}(z),& a < \frac{zr}{r^{2k/n}} 
                \end{cases}.\end{equation}
                where
                \[T^\circ_{i,j,k}(z) := \left[z + 2qr^{2k/n-1}\right]|d_{i,j}|,\ \mbox{and}\]
                \[T^+_{i,j,k}(z) := \left[z + 2(q-1)r^{2k/n}\right]|d_{i,j}|.\]
        \end{lemma}

        The proportional schedules enjoyed the property that the $(f+1)^{st}$ agent to reach any location $x=zd_{i,j}\in I_{i,j}$ was the agent $i+f+1$. The next lemma outlines the conditions for this to also be the case with the generalized schedules.
        \begin{lemma}\label{lm:S_cond}
                Suppose that $a \geq r^{1/n}$. Then agent $i+k$, $k=f+1,\ldots,n$, will reach any position $x = zd_{i,j} \in I_{i,j}$ at the time $T_{i,j,k}(z)=T^\circ_{i,j,k}(z)$. Furthermore, \begin{enumerate}
                        \item if $q \leq \frac{r}{r-r^{2/n}}$ then agents $i+1,\ldots,i+f$ will reach $x$ before agent $i+f+1$
                        \item if $\frac{r}{r-r^{2/n}} < q \leq \frac{r}{r-r^{4/n}}$ then agents $i+1,\ldots,i+f-1$ will reach $x$ before agent $i+f+1$, and agent $i+f$ will reach $x$ before agent $i+f+1$ provided that $a \geq z r^{1/n}$.
                \end{enumerate}
        \end{lemma}

        \begin{lemma}\label{lm:intersects_gen}
                Suppose that $a \geq r^{1/n}$. Then the point $(\rho_{i,j,k},\tau_{i,j,k})$, $k=0,\ldots,f$, at which the trajectory of agent $i$ intersects that of agent $i+k$ while agent $i$ is moving away from $d_{i,j}$ and agent $i+k$ is moving towards $d_{i+k,j}$ is
                \begin{equation}\label{eq:rho_gen}
                        (\rho_{i,j,k},\tau_{i,j,k}) = \begin{cases}
                                (\rho^\circ_{i,j,k},\tau^\circ_{i,j,k}),& a \geq q(r^{1-2k/n}-1)\\
                                (\rho^+_{i,j,k},\tau^+_{i,j,k}),& \mbox{otherwise}
                        \end{cases}
                \end{equation}             
                where
                \begin{align*}                
                        \rho^\circ_{i,j,k} &:= q(1-r^{2k/n-1})d_{i,j},\\
                        \tau^\circ_{i,j,k} &:= q(1+r^{2k/n-1})|d_{i,j}|,
                \end{align*}
                and
                \begin{align*}                
                        \rho^+_{i,j,k} &:= [q-(q-1)r^{2k/n}]d_{i,j},\\
                        \tau^+_{i,j,k} &:= [q+(q-1)r^{2k/n}]|d_{i,j}|.
                \end{align*}   
        \end{lemma}
        Let $\dagger = \circ,+$. Since the ratios $\tau^\dagger_{i,j,k}/|\rho^\dagger_{i,j,k}|$ are independent of $i$ and $j$, for fixed $k$ and $\dagger$ the points $(\rho^\dagger_{i,j,k},\tau^\dagger_{i,j,k})$ all lie along a common cone. Let $\cone_k^\dagger$ refer to the cone with slope $\beta^\dagger_k := \tau_{i,j,k}^\dagger/|\rho^\dagger_{i,j,k}|$ corresponding to the points $(\rho^\dagger_{i,j,k},\tau^\dagger_{i,j,k})$. The proportional schedules had the property that at all times $t>0$ there was an agent between the cones $\cone_0$ and $\cone_1$ and this agent was the most distant from the origin on its respective side. For the generalized schedules we will want the same property to hold for the cones $\cone^+_0$ and $\cone^+_1$. We observe that
        \begin{align}\label{eq:betas}
                \beta^\circ_k &:= \frac{\tau^\circ_{i,j,k}}{|\rho^\circ_{i,j,k}|} = \frac{r+r^{2k/n}}{r-r^{2k/n}},\nonumber\\
                \beta^+_k &:= \frac{\tau^+_{i,j,k}}{|\rho^+_{i,j,k}|} = \frac{q+(q-1)r^{2k/n}}{|q-(q-1)r^{2k/n}|}.
        \end{align}
        We include the absolute value in the denominator of $\beta^+_k$ since it is possible that $q-(q-1)r^{2k/n} < 0$. We will later show that when $k=1$ we will indeed have $q-(q-1)r^{2k/n} \geq 0$ as a result of the condition $q \leq \frac{r}{r-r^{1-2/n}}$ in the definition of $\parset$.

        \begin{lemma}\label{lm:intersects_cond}
                During the time interval $[\tau_{i-1,j,1},\tau_{i,j,1}]$ agent $i$ has the most negative/positive position when $j$ is odd/even provided that 
                $a = q(r^{1-4/n}-1)$
                or
                $q(r^{1-4/n}-1) < a \leq q(r^{1-2/n}-1)$ and $q \leq \frac{r}{r-r^{2/n}}$.
        \end{lemma}

        In the next lemma we describe how we should choose $q$ if we want the agent furthest from $d_{i,j}$ to be located on the cone $\cone_1^+$ at the time the $(f+1)^{st}$ agent reaches $d_{i,j}$.
        \begin{lemma}\label{lm:s_choice}
                If we take $q = \hat{q}(r,k)$ then agents $i+k-1$ and $i+k$, $k=f+3,\ldots,n$, will both be located on the cone $\cone_1^+$ on the opposite side of the origin from $d_{i,j}$ at the time $T^\circ_{i,j,f+1}(1)$.
        \end{lemma}

        We need one last lemma before proving Theorem~\ref{thm:f_crash}.
        \begin{lemma}\label{lm:abound}
                We have $\hat{a}(r,q) > r^{1/n}$ when $r>1$ and $q \geq \frac{r}{r-1}$.
        \end{lemma}
        
        We now have all of the lemmas required to prove Theorem~\ref{thm:f_crash}. The proof itself can be found in the appendix.

        Figure~\ref{fig:traj} in the appendix illustrates the optimized trajectories of the agents for the cases $f=2,3,4$. Table~\ref{tbl:gen} lists the (numerically) optimized competitive ratios of the generalized schedule along with the corresponding optimal pair $(r,u)$ and the resulting parameters $a$ and $q$. Figure~\ref{fig:gen} plots these competitive ratios as a function of $f$ along with those of the optimized proportional schedule. One can observe that, in many cases, the two algorithms have identical competitive ratios. In all cases for $f>1$ the optimal competitive ratio is defined by Scenario A. It is interesting to note that in every case except $f=3$ we have $\hat{a}(r,q) = q(r^{1-2/n}-1)$. In fact, the case $(n,f)=(7,3)$ seems to be unique in a number of ways. It has an identical competitive ratio as the $f=2$ case. Moreover, the optimal parameter values for $q$ and $a$ are identical for $f=3$ and $f=2$, and, although the parameter $r$ is not the same, one can confirm that the quantity $r^{1/n}$ is identical in both cases.
        
\subsection{Three agents, one fault}\label{sec:one_fault}
        Our goal is to prove the following theorem.
        \begin{theorem}\label{thm:31_crash}
                Suppose that $r > 2\sqrt{2}$ and take
                \begin{multline}\label{eq:opt_a1}
                        a = \frac{1}{2(r-2r^{1/3})}\Biggr[r+1+r^{2/3}\\
                        - \sqrt{(r+r^{2/3}+1)^2+4r^{2/3}(r+1)(r-2r^{1/3})}\Biggr],
                \end{multline}        
                \begin{equation}\label{eq:opt_s}
                        q = \frac{r+1-a}{r+1-2r^{1/3}}.
                \end{equation}
                Then the competitive ratio of the generalized schedule for $(n,f)=(3,1)$ is
                \begin{multline}\label{eq:R1}
                        R_1 = 2 + \frac{1}{r+1-2r^{1/3}}\Biggl[r^{1/3}(2-r^{1/3})\\ + \sqrt{(r+1+r^{2/3})^2+4r^{2/3}(r+1)(r-2r^{1/3})}\Biggr].
                \end{multline}
                In particular, if $r = 6.833921$, $a = 1.699557$, and $q = 1.518949$ then $R_1 = 7.437011$.
        \end{theorem}

        Recall the potential worst case Scenario A for the proportional schedule -- the target is placed just beyond a turning point $d_{i,j}$ and agent $i+1$ (who would be the first to reach the target) is faulty. Thus, agent $i+2$ will be the first reliable agent to reach the target and at this time agent $i = i+3 \mod 3$ will still need to evacuate. The problem with the normal proportional schedule is that agent $i$ might be relatively far from location $d_{i,j}$ when it hears the announcement. We will thus use the extra turning points of the generalized algorithm in order to position agent $i$ relatively close to $d_{i,j}$ at the time agent $i+2$ reaches $d_{i,j}$. In particular, if we choose $q$ according to \eqref{eq:opt_s}, then agent $i$ will be located at its sub-turning point $d^{(2)}_{i,j}$ at exactly the same time agent $i+2$ reaches location $d_{i,j}$ (see Lemma~\ref{lm:opt_s}). However, in fixing the value of $q$ in this way we will introduce a new potentially worst-case location of the target and, choosing $a$ according to \eqref{eq:opt_a1}, will ensure this new potential worst case is no worse than Scenario A (see Lemma~\ref{lm:opt_a1}).

        \begin{figure}[!h]
                \centering
                \includegraphics[scale=0.8]{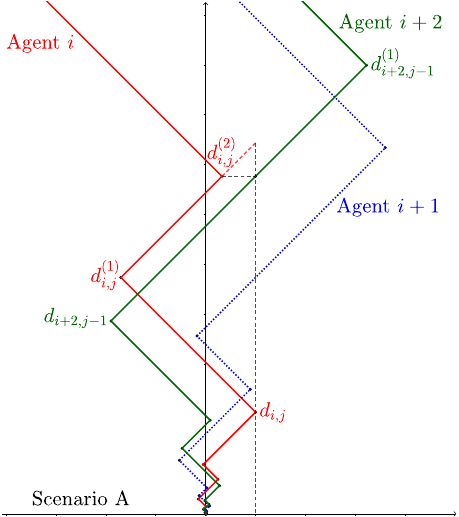}
                \hspace{1cm}
                \includegraphics[scale=0.8]{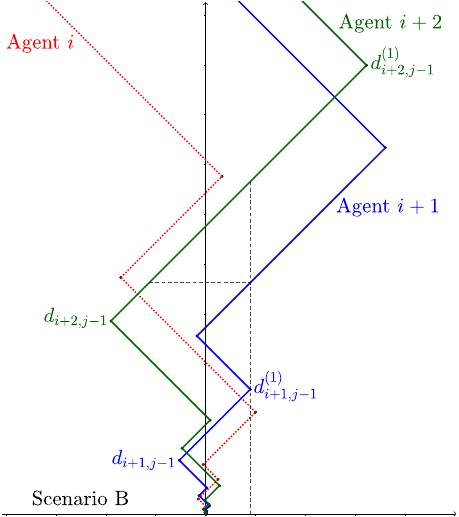}
                \caption{Illustrating the two worst case scenarios for the generalized proportional schedule for $n=3$ and $f=1$ when the parameters $a$ and $q$ are chosen according to \eqref{eq:opt_a1} and \eqref{eq:opt_s}. Scenario A on the left occurs when the target is just beyond the turning point $d_{i,j}$ and agent $i+1$ is faulty. Scenario B on the right occurs when the target is just beyond the turning point $d^{(1)}_{i+1,j-1}$ and agent $i$ is faulty.\label{fig:31_crash}}
        \end{figure}

        Figure~\ref{fig:31_crash} illustrates the trajectories of the three agents when $a$ and $q$ are chosen according to \eqref{eq:opt_a1} and \eqref{eq:opt_s}. The left side of the figure displays Scenario A. One can observe that at the instant agent $i+2$ reaches location $d_{i,j}$ agent $i$ will be at its turning point $d^{(2)}_{i,j}$. The right side of Figure~\ref{fig:31_crash}, which will be referred to as Scenario B, represents the new potentially worst case situation. In this scenario the target is placed at location $d^{(1)}_{i+1,j-1}$ and it is agent $i$ that is faulty. Agent $i+1$ will be the first reliable agent to reach the target and agent $i+2$ will be the final agent to evacuate. Observe that for this scenario the announcement by agent $i+1$ will not affect the trajectory of agent $i+2$ since this agent was already on its way to $d^{(1)}_{i+1,j-1}$ at the time of the announcement. We will choose $a$ in order to ensure that the competitive ratios resulting from Scenarios A and B will be equal.

        \paragraph*{Scenario A: } This is the case depicted on the left of Figure~\ref{fig:31_crash}. The target is just beyond the turning point $d_{i,j}$ and agent $i+2$ is the first reliable agent to reach the target. The next lemma demonstrates that if we choose $s$ according to \eqref{eq:opt_s} then agent $i$ will be at its turning point $d^{(2)}_{i,j}$ at the moment agent $i+2$ reaches $d_{i,j}$.
        \begin{lemma}\label{lm:opt_s}
                If $q$ is given by \eqref{eq:opt_s} and $a \geq \frac{1}{r^{1/3}}$ then agent $i$ will reach its turning point $d_{i,j+2}$ at the same time agent $i+2$ reaches location $d_{i,j}$.
        \end{lemma}

        We now let $R^A_1$ represent the competitive ratio for Scenario A. In the next lemma we derive an expression for this competitive ratio.
        \begin{lemma}\label{lm:RA}
                Suppose that $a \geq \frac{1}{r^{1/3}}$. Then the competitive ratio of Scenario A is
                \[R^A_1 = 1 + \frac{t^{(1)}_{i,j} + |d^{(1)}_{i,j}|}{|d_{i,j}|} = 1+2(q+a).\]
        \end{lemma}

        \paragraph*{Scenario B: } This is the case depicted on the right of Figure~\ref{fig:31_crash}. The target is placed just beyond the turning point $d^{(1)}_{i+1,j-1}$, agent $i+1$ is the first reliable agent to reach the target, and agent $i+2$ will evacuate last. We note that for this case to occur it must be that $|d^{(1)}_{i+1,j-1}| \leq |d_{i,j}|$ which implies that $a \leq r^{1/3}$. Let $R^B_1$ represent the competitive ratio for this scenario.
        \begin{lemma}\label{lm:RB}
                Suppose that $a \leq r^{1/3}$. Then the competitive ratio of Scenario B is
                \[R^B_1 = 1 + \frac{t_{i+2,j-1}+|d_{i+2,j-1}|}{d^{(1)}_{i+1,j-1}} = 1+\frac{2qr^{2/3}}{a}.\]
        \end{lemma}

        \begin{lemma}\label{lm:opt_a1}
                Suppose that $\frac{1}{r^{1/3}} \leq a \leq r^{1/3}$, and $r > 2\sqrt{2}$. Then the competitive ratio of Scenarios A and B will be equal when $q$ is given by \eqref{eq:opt_s} and $a$ is given by \eqref{eq:opt_a1}.
        \end{lemma}
        To derive the expression for $R_1$ in Theorem~\ref{thm:31_crash} we substitute \eqref{eq:opt_a1} and \eqref{eq:opt_s} into the expression for $R^A_1$ in Lemma~\ref{lm:RA}. 
        Numerically optimizing \eqref{eq:R1} with respect to $r > 2\sqrt{2}$ yields the specific results quoted in Theorem~\ref{thm:31_crash}.

\section{Byzantine search with one fault}\label{sec:search}

        In this section we use our evacuation algorithm for $(n,f)=(3,1)$ to improve upon the best known upper bound for {\em search} by three agents, with at most one byzantine fault. For this problem we only require the agents to reliably identify the target and, in particular, they do not need to evacuate. A reliable identification of the target occurs at the first time a provably reliable agent reaches the target. We will prove the following result.
        \begin{theorem}\label{thm:31_search}
                Byzantine search for three agents, at most one of which is faulty, can be completed with a competitive ratio no more than $7.43701137$.
        \end{theorem}

        We first give a quick discussion of the model. A byzantine faulty agent is similar to a crash faulty agent except that byzantine agents can also lie about finding the target, in addition to failing to announce the target. Therefore, in this case, an announcement that the target has been found cannot necessarily be trusted. In order to confirm whether or not the target is at an announced location $x$ at least one provably reliable agent must confirm that the target is at $x$. With $n=2f+1$, it can be required that all agents reach $x$ before this confirmation occurs (which is, of course, just evacuation in disguise). However, in this case it is possible that the agents do not find the target at $x$ and must continue their search. Nevertheless, in this process the reliable agents will be able to identify at least one of the faulty agents (the one who lied) and thus, if the search continues, it will be from a better standpoint in terms of the ratio of faulty agents to reliable agents. In the specific case of $(n,f)=(3,1)$, the moment the single faulty agent is identified, both of the remaining agents will know each other is reliable.

        The proof of Theorem~\ref{thm:31_search} has been moved to the appendix, however, briefly, we demonstrate that the agents can simply use the generalized schedule of Theorem~\ref{thm:31_crash} to achieve the desired competitive ratio. The proof simply checks that the faulty agent cannot achieve a worse competitive ratio by lying about the target's location.

\section{Conclusions}\label{sec:conclusions}

        We have studied the problem of evacuation by $n=2f+1$ agents when at most $f$ of the agents are crash faulty. We introduced a novel type of search algorithm which gives an improvement for the evacuation as compared to the trajectories used for optimal crash search. These trajectories also gave an improvement on the best known upper bound for the problem of search by three agents at most one of which is byzantine faulty.

        Since we did not carry out a pure optimization of the parameters for our new algorithm it is possible that there are better choices of these parameters than the ones described here. Furthermore, it is possible to further generalize the trajectories of the agents with the addition of more sub-turning points, however, this seems rather unlikely to improve the results. It is also an open problem to improve lower bounds for the evacuation problem since the state of the art derives from the optimality of the search problem, i.e. evacuation has at least the same competitive ratio as search since the agents must first find the target in order to evacuate. As a result, there is a rather large gap between the upper bounds presented here and the best known lower bounds. The complexity of the trajectories described here hints that achieving a tight lower bound for this problem may be a difficult task indeed.

        We have focused on the particular case that $n=2f+1$. When $n > 2f+1$ the problem is trivial since we can split the agents into two groups of at least $f+1$ agents each and send them in opposite directions. This strategy easily leads to a competitive ratio of 3, which is tight. The problem is interesting when $n < 2f+1$ and the analysis of these cases offers an extension of this research. There are some simple observations that one can make for these cases. An upper bound of $9$ is clear since we can have all agents stick together and use the doubling strategy for single agent search. However, in most cases it is expected that one can improve upon this upper bound. Indeed, for pairs $(n,f) = (3k, 2k-1)$, $k \geq 1$, we can form 3 groups of $k$ agents and, with at most $2k-1$ faults, two out of three of these groups are guaranteed to contain a reliable agent. We can therefore have each group act as a single agent and perform the generalized schedule algorithm of Theorem~\ref{thm:31_crash} to achieve an upper bound of $\sim 7.437$. A similar trick can be used to achieve the same competitive ratio as achieved for the case $(2f+1,f)$ for all pairs $((2f+1)k, (f+1)k-1)$, $k \geq 1$ and $f \geq 1$. As was the case when $n=2f+1$, lower bounds for $n < 2f+1$ are all due to the optimality of the proportional schedules for crash faulty search.

\bibliographystyle{plain}
\bibliography{refs.bib}

@inproceedings{kupavskii2018lower,
    title={Lower bounds for searching robots, some faulty},
    author={Kupavskii, A. and Welzl, E.},
    booktitle={PODC 2018},
    pages={447--453},
    year={2018},
    address={Egham, United Kingdom},
    publisher={ACM}
}

@article{AH00,
  title={Exploring unknown environments},
  author={Albers, S. and Henzinger, M. R.},
  journal={SIAM Journal on Computing},
  volume={29},
  number={4},
  pages={1164--1188},
  year={2000},
  publisher={SIAM}
}

@article{albers2002exploring,
  title={Exploring unknown environments with obstacles},
  author={Albers, S. and Kursawe, K. and Schuierer, S.},
  journal={Algorithmica},
  volume={32},
  number={1},
  pages={123--143},
  year={2002},
  publisher={Springer}
}

@inproceedings{DKP91,
  title={How to learn an unknown environment},
  author={Deng, X. and Kameda, T. and Papadimitriou, C.},
  booktitle={FOCS 1991},
  pages={298--303},
  year={1991},
  organization={IEEE}
}

@article{HIKK01,
  title={The polygon exploration problem},
  author={Hoffmann, F. and Icking, C. and Klein, R. and Kriegel, K.},
  journal={SIAM Journal on Computing},
  volume={31},
  number={2},
  pages={577--600},
  year={2001},
  publisher={SIAM}
}

@article{bellman1963optimal,
    title={An optimal search},
    author={Bellman, R.},
    journal={SIAM Review},
    volume={5},
    number={3},
    pages={274--274},
    year={1963},
    publisher={SIAM}
}

@article{beck1964linear,
    title={On the linear search problem},
    author={Beck, A.},
    journal={Israel J. of Mathematics},
    volume={2},
    number={4},
    pages={221--228},
    year={1964},
    publisher={Springer}
}

@article{beck1965more,
  title={More on the linear search problem},
  author={Beck, A.},
  journal={Israel J. of Mathematics},
  volume={3},
  number={2},
  pages={61--70},
  year={1965},
  publisher={Springer}
}

@article{beck1970yet,
    title={Yet more on the linear search problem},
    author={Beck, A. and Newman, D.},
    journal={Israel J. of Mathematics},
    volume={8},
    number={4},
    pages={419--429},
    year={1970},
    publisher={Springer}
}

@article{fristedt_heath_1974,
    title={Searching for a particle on the real line}, 
    volume={6}, 
    DOI={10.2307/1426208}, 
    number={1}, 
    journal={Advances in Applied Probability}, 
    publisher={Cambridge University Press}, 
    author={Fristedt, B. and Heath, D.}, 
    year={1974}, 
    pages={79–102}
}

@article{fristedt1977hide,
    title={Hide and seek in a subset of the real line},
    author={Fristedt, B.},
    journal={International Journal of Game Theory},
    volume={6},
    number={3},
    pages={135--165},
    year={1977},
    publisher={Springer}
}

@article{gal1972general,
    title={A general search game},
    author={Gal, S.},
    journal={Israel Journal of Mathematics},
    volume={12},
    number={1},
    pages={32--45},
    year={1972},
    publisher={Springer}
}

@article{baezayates1993searching,
    title={Searching in the plane},
    author={Baeza-Yates, R. and Culberson, J. and Rawlins, G.},
    journal={Information and Computation},
    volume={106},
    number={2},
    pages={234--252},
    year={1993},
    publisher={Elsevier}
}

@article{baezayates1995parallel,
    title={Parallel searching in the plane},
    author={Baeza-Yates, R. and Schott, R.},
    journal={Computational Geometry},
    volume={5},
    number={3},
    pages={143--154},
    year={1995},
    publisher={Elsevier}
}

@inproceedings{chrobak2015group,
    title       ={Group Search on the Line},
    author      ={Chrobak, M. and G{\k{a}}sieniec, L. and Gorry T. and R. Martin},
    booktitle   ={SOFSEM 2015},
    address     ={Sn\v{e}\v{z}kou, Czech Republic},
    year        ={2015},
    publisher   ={Springer},
    pages       ={164--176}
}

@article{bampas2019linear,
    title={Linear search by a pair of distinct-speed robots},
    author={Bampas, E. and Czyzowicz, J. and G{\k{a}}sieniec, L. and Ilcinkas, D. and Klasing, R. and Kociumaka, T. and Paj{\k{a}}k, D.},
    journal={Algorithmica},
    volume={81},
    number={1},
    pages={317--342},
    year={2019},
    publisher={Springer}
}

@article{demaine2006online,
  title={Online searching with turn cost},
  author={Demaine, E.D. and Fekete, S.P. and Gal, S.},
  journal={Theoretical Computer Science},
  volume={361},
  number={2},
  pages={342--355},
  year={2006},
  publisher={Elsevier}
}

@inproceedings{CGKKKLNOS19wireless,
    author      ={Czyzowicz, J. and Georgiou, K. and Killick, R. and Kranakis, E. and Krizanc, D. and Lafond, M. and Narayanan, L. and Opatrny, J. and Shende, S.},
    title       ={Energy Consumption of Group Search on a Line},
    booktitle   ={ICALP 2019},
    pages       ={137:1--137:15},
    year        ={2019},
    address     ={Patras, Greece},
    publisher   ={LIPIcs}
}

@article{czyzowicz2021time,
  title={Time-energy tradeoffs for evacuation by two robots in the wireless model},
  author={Czyzowicz, J. and Georgiou, K. and Killick, R. and Kranakis, E. and Krizanc, D. and Lafond, M. and Narayanan, L. and Opatrny, J. and Shende, S.},
  journal={Theoretical Computer Science},
  volume={852},
  pages={61--72},
  year={2021},
  publisher={Elsevier}
}

@inproceedings{PODC16,
    title       ={Search on a Line with Faulty Robots},
    author      ={Czyzowicz, J. and Kranakis, E. and Krizanc, D. and Narayanan, L. and Opatrny J.},
    booktitle   ={PODC 2016},
    address     ={Chicago, Illinois},
    pages       ={405--414},
    year        ={2016},
    publisher   ={ACM}
}

@inproceedings{isaacCzyzowiczGKKNOS16,
    author    = {Czyzowicz, J. and
                Georgiou, K. and
                Kranakis, E. and
                Krizanc, D. and
                Narayanan, L. and
                Opatrny, J. and
                S. Shende},
    title     = {Search on a Line by Byzantine Robots},
    booktitle = {ISAAC 2016},
    address   = {Toronto, Canada},
    pages     = {27:1--27:12},
    year      = {2016},
    publisher = {LIPIcs}
}

@incollection{Sun2020,
  title={Better Upper Bounds for Searching on a Line with Byzantine Robots},
  author={Sun, X. and Sun, Y. and Zhang, J.},
  booktitle={Complexity and Approximation},
  pages={151--171},
  year={2020},
  publisher={Springer}
}

@inproceedings{inbookryan,
author = {Czyzowicz, J. and Killick, R. and Kranakis, E. and Stachowiak, G.},
year = {2021},
month = {January},
pages = {217-227},
booktitle = {SIAM Conference on Applied and Computational Discrete Algorithms (ACDA21)},
title = {Search and evacuation with a near majority of faulty agents},
publisher = {SIAM},
isbn = {978-1-61197-683-0},
doi = {10.1137/1.9781611976830.20}
}

\appendix

\section{Proofs missing from Subsection~\ref{sec:prelim}}
\begin{proof}(Lemma~\ref{lm:wc})
        Consider an evacuation algorithm defined by turning point trajectories. Suppose that the competitive ratio of this algorithm is $\alpha$ and consider a location $x$ that defines this ratio. We will make use of the work of \cite{kupavskii2018lower} which states that we have $\alpha > 3$.
        
        Without loss of generality assume that $x>0$, and, for the sake of deriving a contradiction, suppose that $x$ is not a turning point of one of the agents. Since $x$ is not at a turning point there exists $u>0$ such that the interval $[x-u,x]$ does not contain any turning points. 
        
        Let $S$ represent the first time the target at $x$ is announced; let $i_*$ represent the identity of the agent most distant from $x$ at time $S$; and let $\Delta$ represent the distance between agent $i_*$ and $x$ at the time $S$. Similarly define $S'$, $i_*'$, and $\Delta'$ for the situation that the target is at $x'=x-u$ instead of $x$. By assumption we have
        \[\alpha = \frac{S+\Delta}{x} > \frac{S'+\Delta'}{x-u}.\]
        Consider how $\Delta$ and $S$ would change if the target were instead placed at $x'=x-u$. The agent that announces the location $x$ must have been moving in the positive direction when it reached $x$. Since agents move at constant unit speed between turning points it is clear that $S' = S - u$. For $\Delta'$ there are a few cases to consider. If we have $i_*=i_*'$ and agent $i_*$ was moving towards $x$ at time $S$ then it is clear that $\Delta' = \Delta$. If $i_*=i_*'$ and agent $i_*$ is moving away from $x$ at time $S$ then $\Delta' = \Delta-2u$. If $i_* \neq i_*'$ agents $i_*$ and $i_*'$ crossed paths during the time interval $[S',S]$. In any case, it is simple to see that we will have $\Delta-2u \leq \Delta' \leq \Delta$. We can thus conclude that
        \begin{multline*}
                \frac{S'+\Delta'}{x-u} \geq \frac{S+\Delta-3u}{x-u} = \frac{\alpha x - 3u}{x-u}\\
                 = \frac{\alpha(x-u) + (\alpha-3)u}{x-u} = \alpha + \frac{(\alpha-3)}{x-u} > \alpha
        \end{multline*}
        where the last step results from the fact that $\alpha > 3$. We have therefore arrived at a contradiction and must conclude that the lemma holds.
\end{proof}

\section{Proofs and figures missing from Section~\ref{sec:crash_evac}}

\begin{figure}[tbhp]
        \centering
        \includegraphics[scale=0.19]{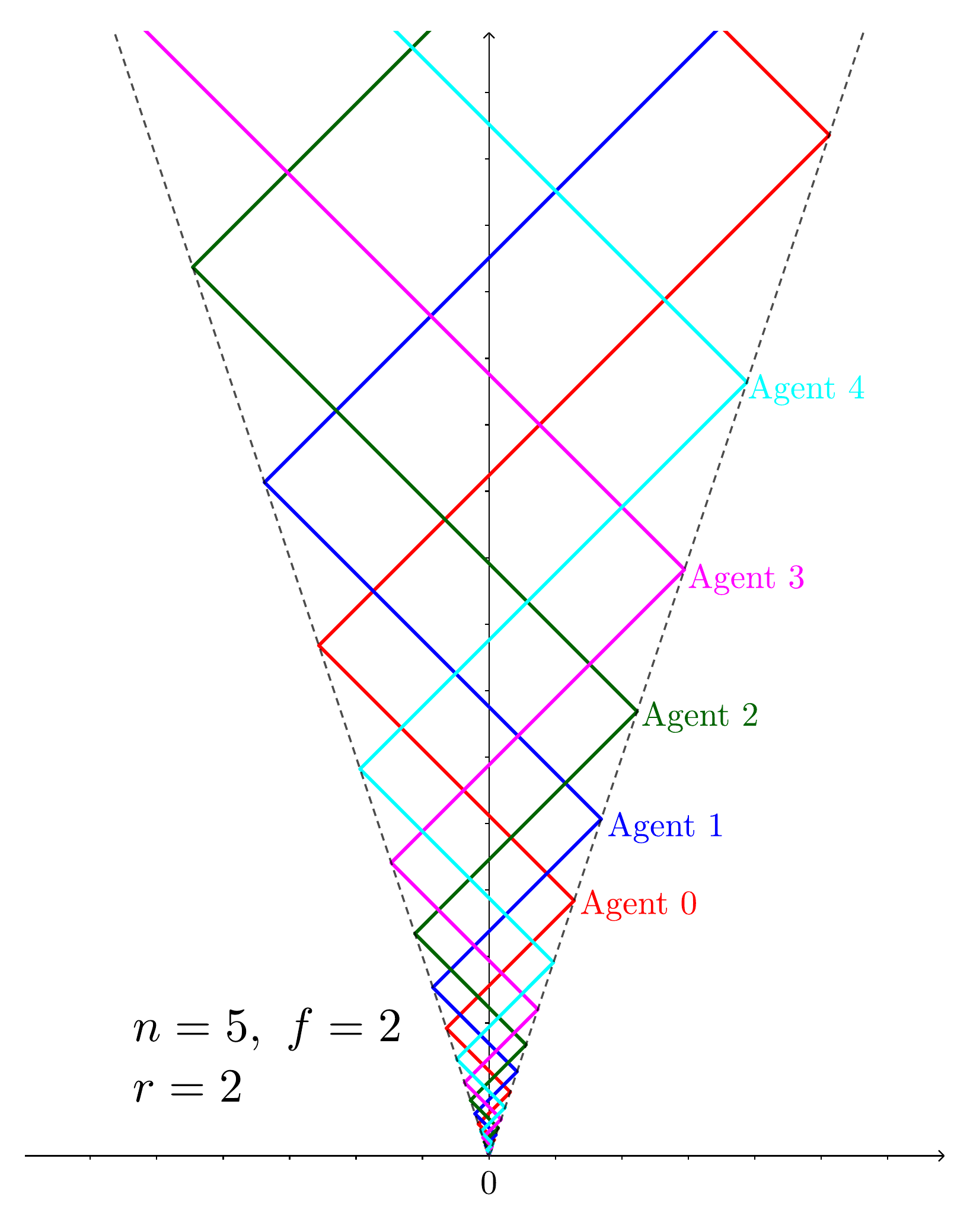}\\
        \includegraphics[scale=0.19]{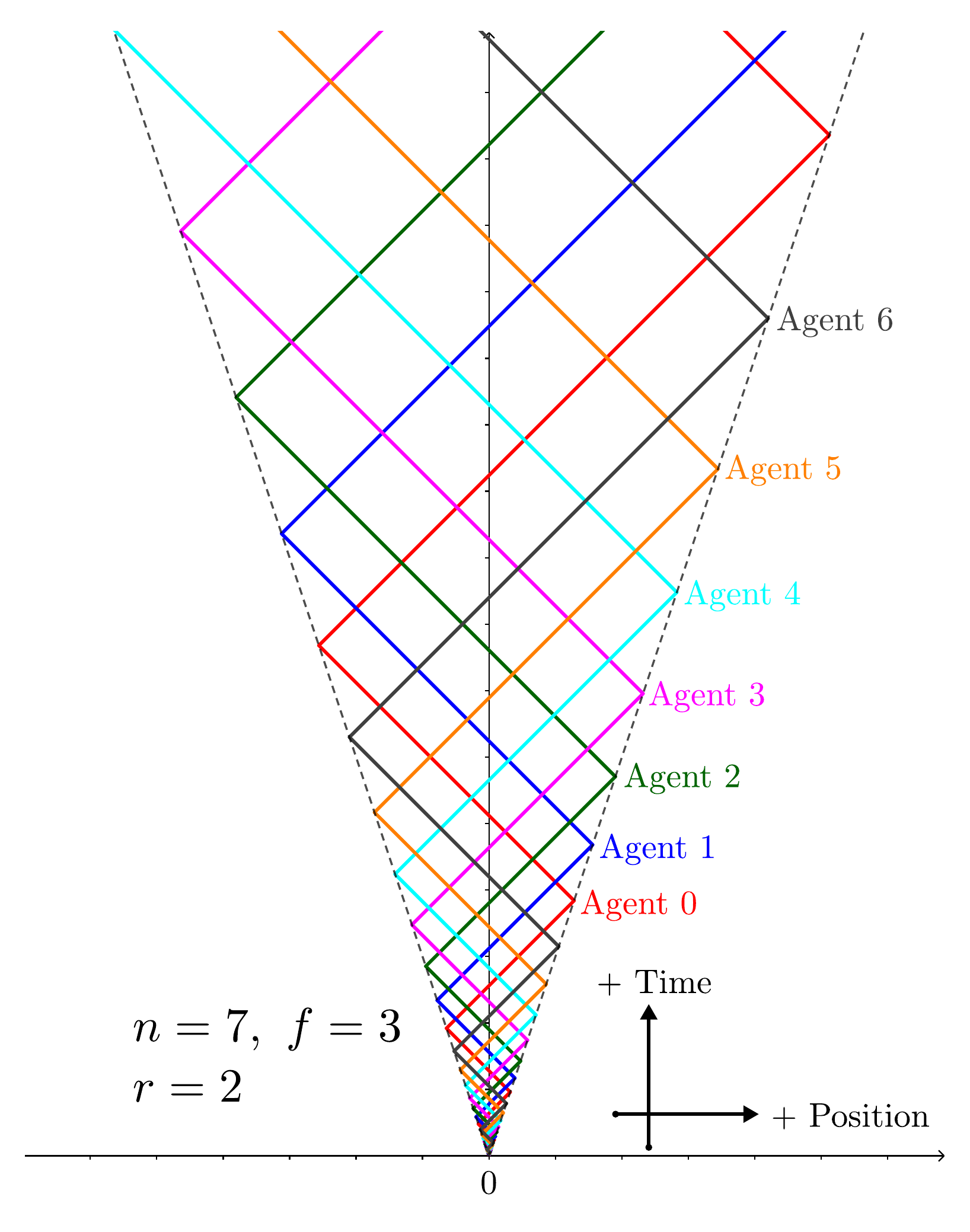}\\
        \includegraphics[scale=0.19]{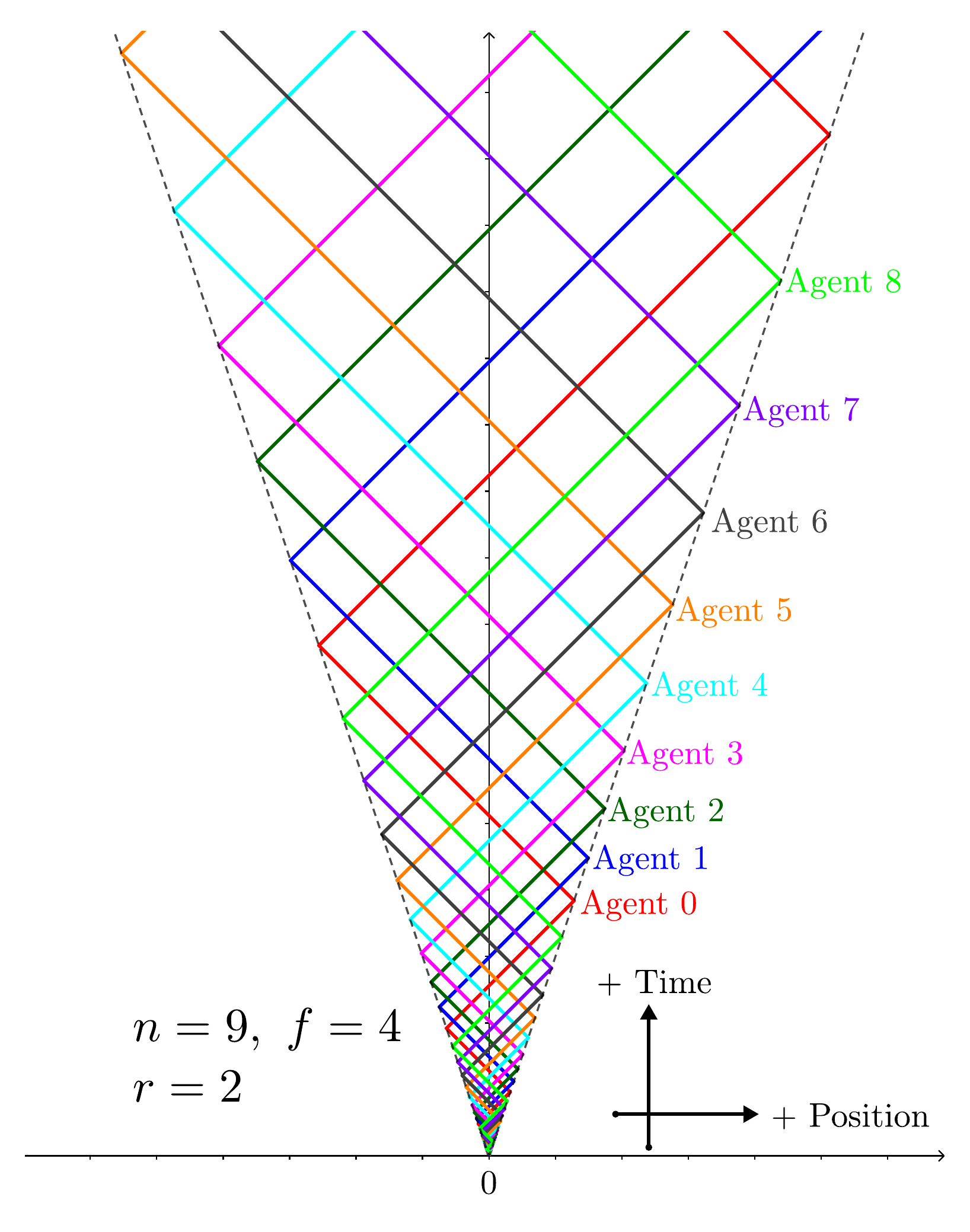}        
        \caption{Example space-time trajectories when $r=2$ and $n=5,7,9$. Note that each of the turning-points lie along a cone with slope $\pm \frac{r+1}{r-1}$.}\label{fig:example_prop}
\end{figure}

\begin{proof}(Lemma~\ref{lm:tij})
        Agent $i$ will reach $d_{i,j-1}$ at time $t_{i,j-1}$ and will then travel distance $|d_{i,j}-d_{i,j-1}|$ at unit speed to reach location $d_{i,j}$. We thus have the recursion
        \[t_{i,j} = t_{i,j-1} + |d_{i,j}-d_{i,j-1}| = t_{i,j-1} + (r+1)|d_{i,j-1}|.\]
        Unrolling this recursion leads to
        \begin{align*}
                t_{i,j} &= (r+1)\sum_{k=-\infty}^{j-1}|d_{i,k}|
                = (r+1)r^{2i/n}\sum_{k=-\infty}^{j-1}r^k\\
                &= \frac{r+1}{r-1} r^{2i/n + j}
                = \frac{r+1}{r-1} |d_{i,j}|.
        \end{align*}
\end{proof}

\begin{proof}(Lemma~\ref{lm:Tijk})
        Since $d_{i,j} \not \in I_{i,j}$ agent $i$ will not reach $x$ until after its turning point $d_{i,j+1}$. On the other hand, agent $i+1$ will visit $x$ while traveling towards $d_{i+1,j}$ and it is clear that this will be the first visit of agent $i+1$ to $x$. More generally, agent $i+k$, $k>0$, will reach $x$ for the first time while traveling towards $d_{i+k,j}$. We can thus conclude that
        \begin{align*}
                T_{i,j,k}(z) &= t_{i+k,j} - |d_{i+k,j}-z d_{i,j}|\\
                &= t_{i+k,j} - |d_{i+k,j}| + z|d_{i,j}|\\
                &= \frac{r+1}{r-1}|d_{i+k,j}| - |d_{i+k,j}| + z|d_{i,j}|\\
                &= \frac{2}{r-1}|d_{i+k,j}| + z|d_{i,j}|\\
                &= \left(z + \frac{2r^{2k/n}}{r-1}\right)|d_{i,j}|.
        \end{align*}
\end{proof}

\begin{proof}(Lemma~\ref{lm:rho_tau})
        One can observe that while travelling away from its turning point $j$ agent $i$ will be moving in the direction $-(-1)^j$ along a line with equation
        \[t = t_{i,j}-(-1)^j(x-d_{i,j})\]
        On the other hand, agent $i+k$ will be traveling in the direction $(-1)^j$ along the line
        \[t = t_{i+k,j} + (-1)^j(x-d_{i+k,j})\]
        while moving towards its turning point $d_{i+k,j}$. The position $\rho_{i,j,k}$ at which agent $i$ and $i+k$ meet can thus be determined by subtracting one of these lines from the other and solving for $x$. We find that
        \begin{align*}
                &\rho_{i,j,k} = \frac{(-1)^j}{2}\left[(-1)^j(d_{i,j}+d_{i+k,j}) + (t_{i,j} - t_{i+k,j})\right]\\
                &= \frac{(-1)^j}{2}\left[(|d_{i,j}|+|d_{i+k,j}|) + \frac{r+1}{r-1}(|d_{i,j}| - |d_{i+k,j}|)\right]\\
                &= \frac{(-1)^j|d_{i,j}|}{2}\left[(1+r^{2k/n}) + \frac{r+1}{r-1}(1-r^{2k/n})\right]\\
                &= \frac{d_{i,j}}{2(r-1)}\left[(r-1)(1+r^{2k/n}) + (r+1)(1-r^{2k/n})\right]\\
                &= \frac{d_{i,j}}{2(r-1)}\left[2r - 2r^{2k/n}\right]
                = \frac{r-r^{2k/n}}{r-1}d_{i,j}.
        \end{align*}
        The time $\tau_{i,j}$ is then
        \begin{align*}
                \tau_{i,j,k} &= -(-1)^j(\rho_{i,j,k}-d_{i,j})+t_{i,j}\\
                &= -|d_{i,j}|\left(\frac{r-r^{2k/n}}{r-1}-1\right)+\frac{r+1}{r-1}|d_{i,j}|\\
                &= \frac{r+1+r-1-r+r^{2k/n}}{r-1}|d_{i,j}|\\
                &= \frac{r+r^{2k/n}}{r-1}|d_{i,j}|.
        \end{align*}
        Note that when $k=0$ we have $\rho_{i,j,0}=d_{i,j}$ and $\tau_{i,j,0} = t_{i,j}$ and thus this ``self intersection'' point is just a turning point.
\end{proof}

\begin{figure}[tbhp]
        \centering
        \includegraphics[scale=0.2]{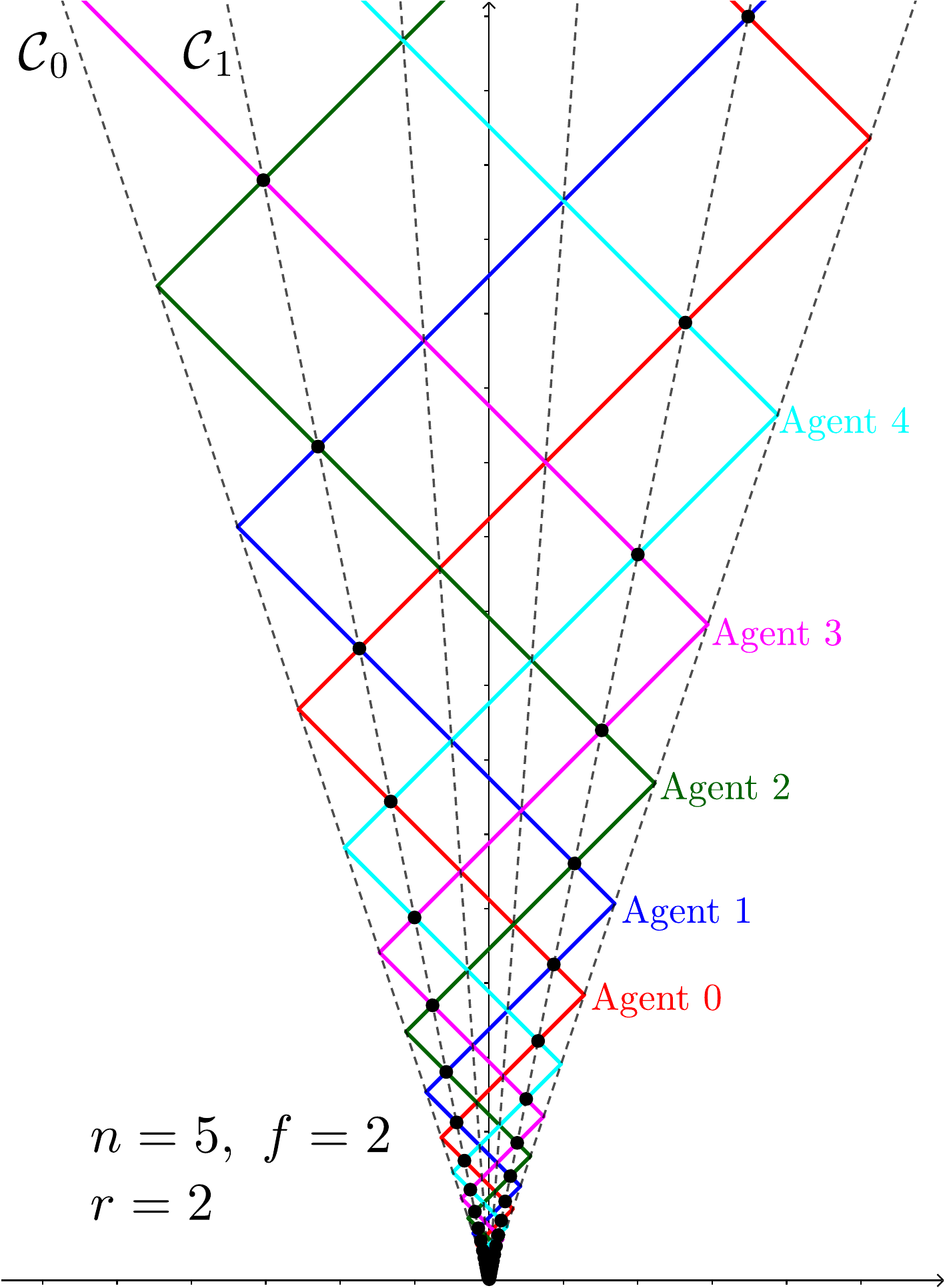}
        \includegraphics[scale=0.2]{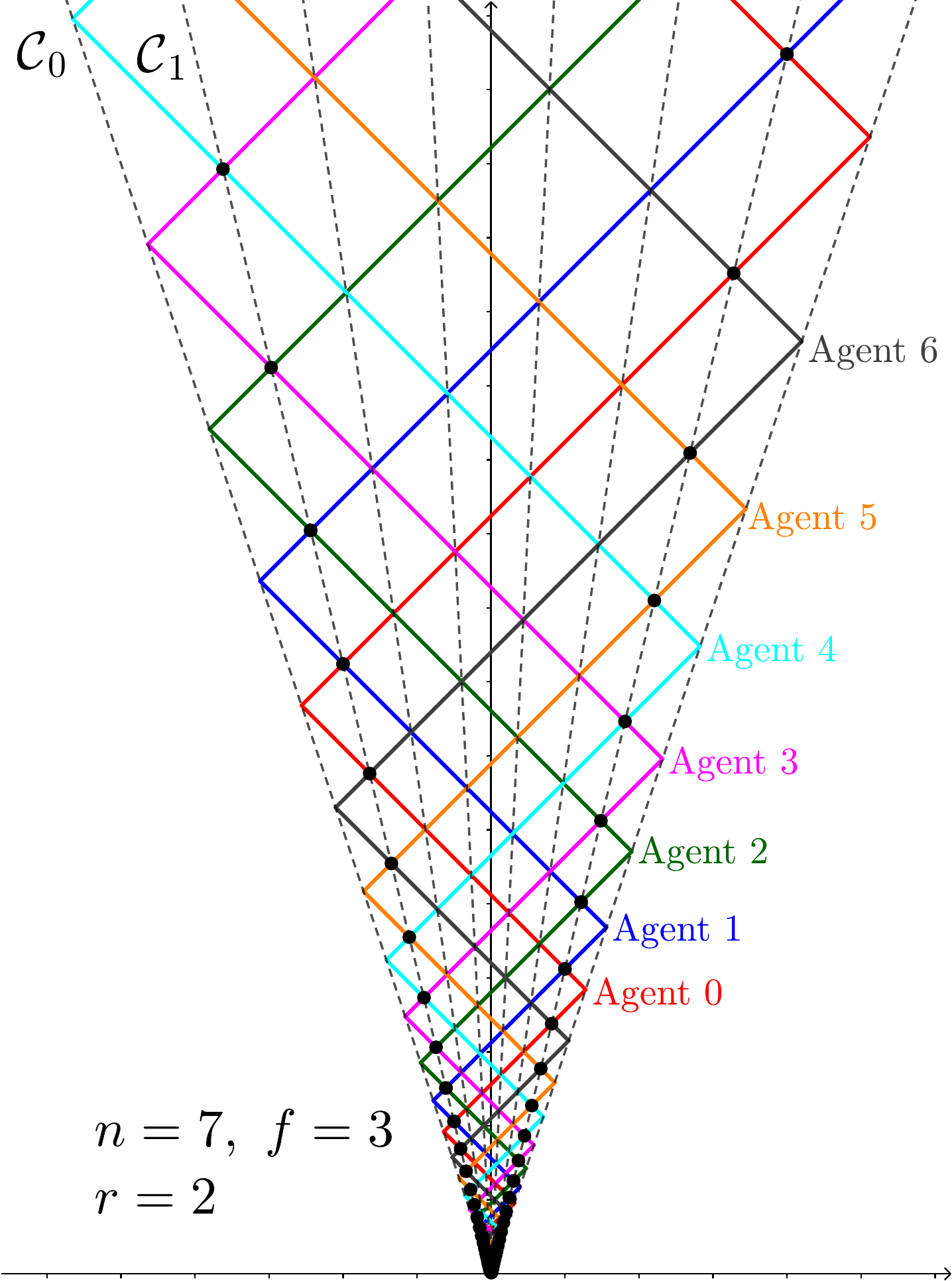}
        \includegraphics[scale=0.2]{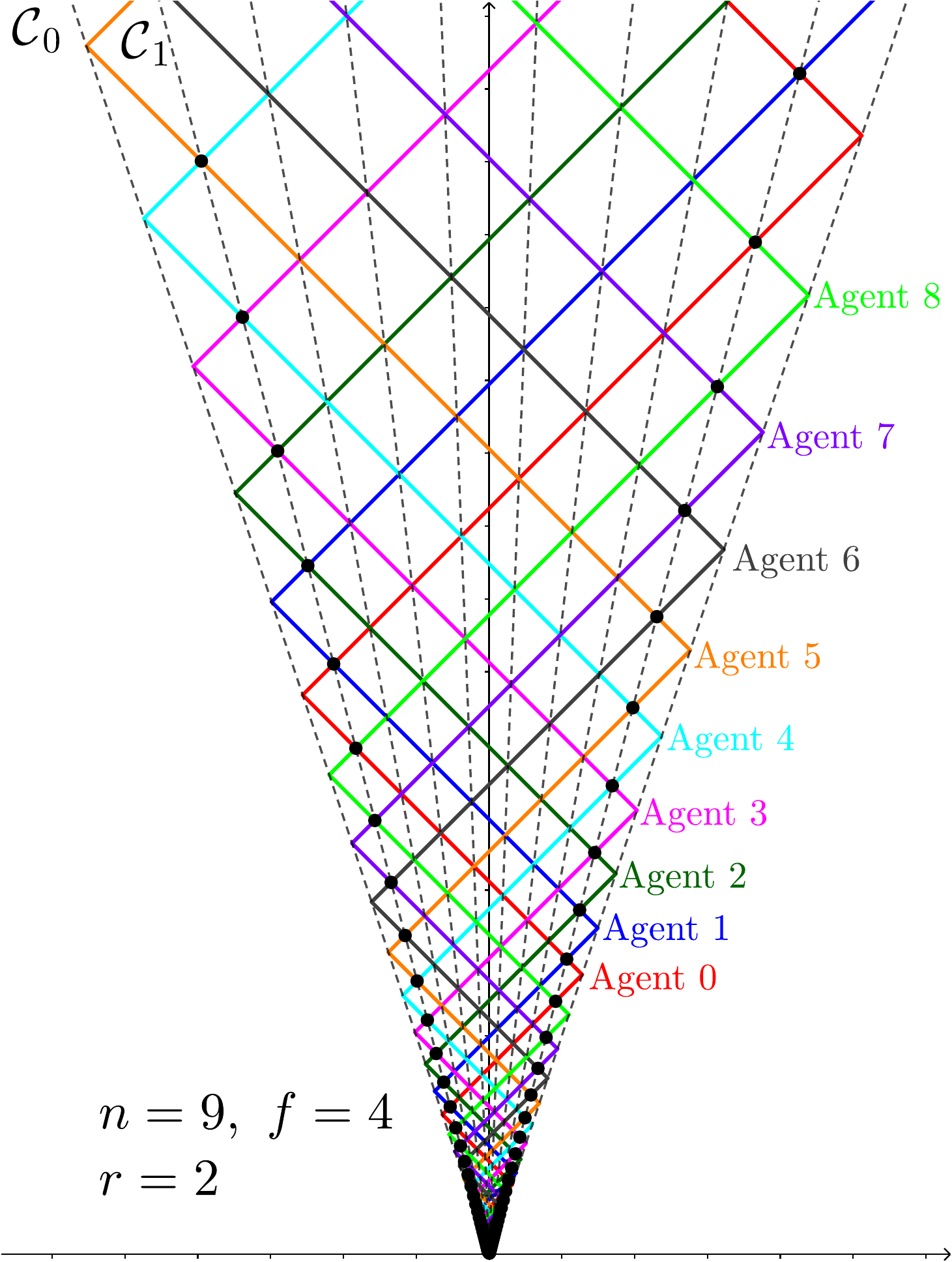}
        \caption{Illustrating the points $(\rho_{i,j,k},\tau_{i,j,k})$ and their respective cones. Only $(\rho_{i,j,1},\tau_{i,j,1})$ are indicated.\label{fig:intersects}}
\end{figure}

\begin{proof}(Theorem~\ref{thm:crash_evac_prop})
        In the worst case the target is just beyond a turning point and so we assume that the target is at location $x = (1+\eps) d_{i,j}$, with $\eps > 0$ arbitrarily small. Also suppose that it is agent $k \in \{1,\ldots,f+1\}$ that is the first reliable agent to reach the target. Then by Lemma~\ref{lm:Tijk}, Lemma~\ref{lm:delta}, and equation \eqref{eq:evac_gen} we have
        \begin{align*}
                E_f^x &\geq |x| + T_{i,j,k}(1+\eps)\left(1 + \frac{1}{\beta_1}\right)\\
                E_f^x &\leq |x| + T_{i,j,k}(1+\eps)\left(1 + \frac{1}{\beta_0}\right)
        \end{align*}
        Dividing by $|x|=(1+\eps)|d_{i,j}|$, substituting in the expressions for $T_{i,j,k}(1)$, $\beta_0$, and $\beta_1$ then yields
        \begin{align*}
                R_f &\geq 1 + \frac{2r}{r+r^{2/n}} + \frac{4r^{1+2k/n}}{(1+\eps)(r+r^{2/n})(r-1)}\\
                R_f &\leq 1 + \frac{2r}{r+1} + \frac{4r^{1+2k/n}}{(1+\eps)(r+1)(r-1)}
        \end{align*}           
        the right hand side of both inequalities increases with $k$ and so taking $k=f+1$ yields  
        \begin{align*}
                R_f &\geq 1 + \frac{2r}{r+r^{2/n}} + \frac{4r^{2+1/n}}{(1+\eps) (r+r^{2/n})(r-1)}\\
                R_f &\leq 1 + \frac{2r}{r+1} + \frac{4r^{2+1/n}}{(1+\eps)(r+1)(r-1)}.
        \end{align*}       
        The theorem then follows by taking $\eps \rightarrow 0$ on the right hand side.     
\end{proof}

\begin{proof}(Theorem~\ref{thm:asymp})
        Following from the discussion preceding Theorem~\ref{thm:asymp} we have
        \begin{align*}
                \hat{R} &= \lim_{f \rightarrow \infty}  \left[1+\frac{2r}{r+1} + \frac{4r^{2+1/n}}{(r+1)(r-1)}\right]\\
                &= 1+\frac{2r}{r+1} + \frac{4r^2}{(r+1)(r-1)}\\
                &= 1 + \frac{6r^2-2r}{r^2-1}
                = 7 - \frac{2(r-3)}{r^2-1}.
        \end{align*}
        The second part of the lemma follows by optimizing the asymptotic competitive ratio with respect to $r$. Observe that
        \[\frac{d\hat{R}}{dr} = \frac{4r(r-3)}{(r^2-1)^2} - \frac{2}{r^2-1}.\]
        Setting this equal to zero and rearranging yields the quadratic equation
        \[r^2-6r+1=0\]
        which can be solved to find $r = 3 \pm 2\sqrt{2}$. Taking the positive root (since $r>1$) and substituting this into the expression for $\hat{R}$ yields
        \begin{align*}
                \hat{R} &= 7 - \frac{4\sqrt{2}}{(3+2\sqrt{2})^2-1} = 7 - \frac{4\sqrt{2}}{16+12\sqrt{2}}\\
                &= 7 - \frac{1}{2\sqrt{2}+3}
                = 7 - \frac{3-2\sqrt{2}}{(2\sqrt{2}+3)(3-2\sqrt{2})}\\
                &= 7 - (3-2\sqrt{2}) = 4+2\sqrt{2}.
        \end{align*}
\end{proof}

\begin{figure}[tbhp]
        \centering
        \includegraphics[width=\linewidth,keepaspectratio]{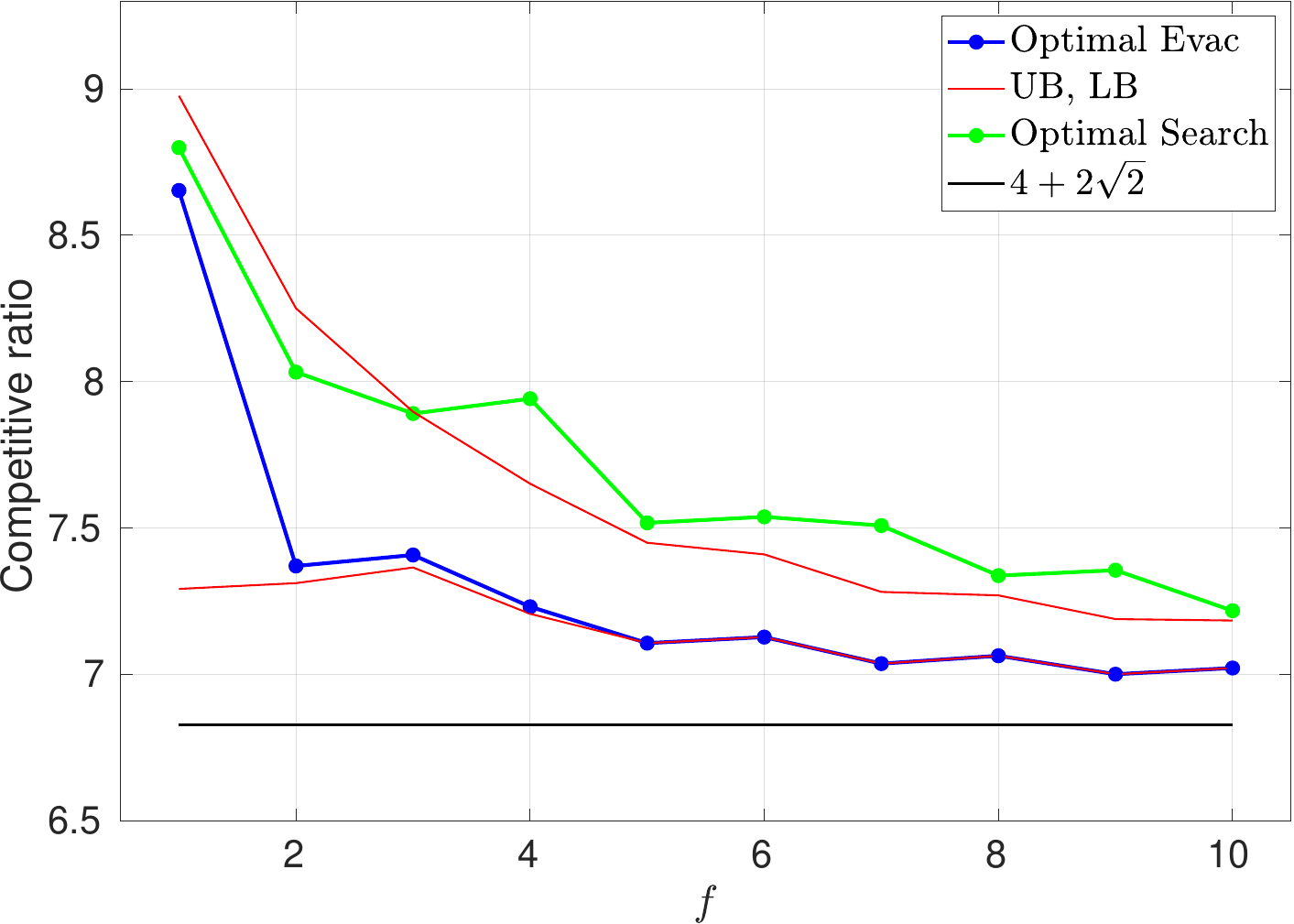}
        \caption{The competitive ratio as a function of the number of faults $f$ for the proportional schedule algorithm. The blue line indicates the optimized competitive ratio; the red lines indicate the bounds of Theorem~\ref{thm:crash_evac_prop}; the green line indicates the competitive ratio of evacuation when we choose $r$ to optimize only $S_f^x$; the lower black line indicates the asymptotic limit of $4+2\sqrt{2}$.\label{fig:cr_compare}}
\end{figure}

\section{Proofs and figures missing from Section~\ref{sec:gen_prop_schedules}}
\subsection{Proofs and figures missing from Subsection~\ref{sec:many_faults}}
\begin{proof}(Lemma~\ref{eq:tij_ell})
        While traveling from $d_{i,j-1}$ to $d_{i,j}$ agent $i$ will travel distance $|d_{i,j}-d_{i,j-1}|$ plus twice the distance between $d^{(1)}_{i,j-1}$ and $d^{(2)}_{i,j-1}$. We thus have
        \begin{align*}
                t_{i,j} &= t_{i,j-1} + |d_{i,j}-d_{i,j-1}|+2|d^{(2)}_{i,j-1}-d^{(1)}_{i,j-1}|\\
                &= t_{i,j-1} + (r+1)|d_{i,j-1}| + 2s|d_{i,j-1}|\\
                &= t_{i,j-1} + (r+1+2s)|d_{i,j-1}|\\
                &= (r+1+2s)\sum_{k=-\infty}^{j-1}|d_{i,k}|\\
                &= \frac{r+1+2s}{r-1}|d_{i,j}| = (2q-1)|d_{i,j}|.
        \end{align*}
        After reaching $d_{i,j}$ at time $t_{i,j}$ agent $i$ must travel distance $|d^{(1)}_{i,j}-d_{i,j}| = (1+a)|d_{i,j}|$ to reach $d^{(1)}_{i,j}$ and thus $t^{(1)}_{i,j} = t_{i,j}+(1+a)|d_{i,j}| = (2q+a)|d_{i,j}|$. Similarly, after reaching $d^{(1)}_{i,j}$ at time $t^{(1)}_{i,j}$ agent $i$ must travel distance $|d^{(2)}_{i,j}-d^{(1)}_{i,j}| = s |d_{i,j}|$ to reach $d^{(2)}_{i,j}$. Thus $t^{(2)}_{i,j} = t^{(1)}_{i,j} + s|d_{i,j}| = t_{i,j} + (1+s+a)|d_{i,j}| = (2q+s+a)|d_{i,j}| = [q(r+1)-r+a]|d_{i,j}|$.
\end{proof}

\begin{proof}(Lemma~\ref{lm:tijk_gen})
        As was the case in Lemma~\ref{lm:Tijk}, agent $i+k$, $k=1,\ldots,n$, will arrive to location $x$ while traveling between its turning points $d_{i+k,j-1}$ and $d_{i+k,j}$. However, for the generalized algorithm, the expression for $T_{i,j,k}(z)$ will depend on whether or not agent $i+k$ reaches $x$ before or after its sub-turning point $d^{(1)}_{i+k,j-1}$, i.e. whether or not we have $x \in [d_{i+k,j-1},d^{(1)}_{i+k,j-1}]$ or $x \in (d^{(1)}_{i+k,j-1},d_{i+k,j}]$. We consider first the case that $x \in [d_{i+k,j-1},d^{(1)}_{i+k,j-1}]$. We have
        \begin{align*}
                T_{i,j,k}(z) &= t_{i+k,j-1} + |z d_{i,j}-d_{i+k,j-1}|\\
                &= t_{i+k,j-1} + z|d_{i,j}|+|d_{i+k,j-1}|\\
                &= 2q|d_{i+k,j-1}| + z|d_{i,j}|\\
                &= \left[z+2qr^{2k/n-1}\right]|d_{i,j}|
        \end{align*}
        If $x \in (d^{(1)}_{i+k,j-1},d_{i+k,j}]$ then agent $i+k$ will travel an extra distance equal to $2|d^{(2)}_{i+k,j-1}-d^{(1)}_{i+k,j-1}| = 2s|d_{i+k,j-1}|$ to reach $x$ as compared to the previous case. We thus have
        \begin{align*}
                T_{i,j,k}(z) &= \left[z + 2qr^{2k/n-1}\right]|d_{i,j}| + 2s|d_{i+k,j-1}|\\
                &= \left[z + 2(q+s)r^{2k/n-1}\right]|d_{i,j}|\\
                &= \left[z + 2(q-1)r^{2k/n}\right]|d_{i,j}|.
        \end{align*}
        We now establish under what conditions we find ourselves in each of these two cases. To do this we need to compare $x$ to $d^{(1)}_{i+k,j-1}$. To have $x \in [d_{i+k,j-1},d^{(1)}_{i+k,j-1}]$ when $j$ is even we need $x \leq d_{i+k,j-1}$ and when $j$ is odd we need $x \geq d_{i+k,j-1}$. We thus need to consider the inequality
        \begin{align*}
                &(-1)^j x \leq (-1)^j d^{(1)}_{i+k,j-1}\\
                \rightarrow \quad& z |d_{i,j}| \leq a r^{2k/n-1}|d_{i,j}|\\
                \rightarrow \quad& a \geq \frac{z\,r}{r^{2k/n}}.
        \end{align*}
\end{proof}

\begin{proof}(Lemma~\ref{lm:S_cond})
        By Lemma~\ref{lm:tijk_gen} we will have $T_{i,j,f+1}(z) = T^\circ_{i,j,f+1}(z)$ when $a \geq \frac{zr}{r^{2(f+1)/n}} = \frac{z}{r^{1/n}}$. For $x \in I_{i,j}$ we have $z \in (1,r^{2/n}]$ and thus agent $i$ will reach any position $x \in I_{i,j}$ provided that $a \geq \frac{r^{2/n}}{r^{1/n}}=r^{1/n}$.

        Now suppose that $a \geq r^{1/n}$. Since the condition $a \geq \frac{zr}{r^{2k/n}}$ gets easier to satisfy for larger $k$ we can conclude that agents $i+k$, $k=f+2,\ldots,n$, will reach $x$ at times $T_{i,j,k}(z) = T^\circ_{i,j,k}(z)$. Let $\dagger=\circ,+$ and observe that for fixed $i$, $j$, and $z$ we have
        \[T^\dagger_{i,j,1}(z) < T^\dagger_{i,j,2}(z) < \ldots < T^\dagger_{i,j,f+1}(z).\]
        In particular, agents $i+f+2,\ldots,i+n$ will arrive to $x$ after agent $i+f+1$. Thus, in order to be the $(f+1)^{st}$ agent to reach $x$, agents $i+1,i+2,\ldots,i+f$ must reach $x$ before agent $i+f+1$.

        Define $k_*$ as the largest integer such that $a < \frac{zr}{r^{2k_*/n}}$. Then for each $k \leq k_*$ we have $T_{i,j,k}(z) = T^+_{i,j,k}(z)$ and agent $i+k_*$ will arrive to $x$ after agents $i+1,\ldots,i+k_*-1$. Moreover, agents $i+k_*+1,\ldots,i+f$ will arrive to $x$ before agent $i+f+1$. Thus, to make sure agent $i+f+1$ is the $(f+1)^{st}$ agent to reach $x$ we must make sure that $T^\circ_{i,j,f+1}(z) \geq T^+_{i,j,k_*}(z)$. We have
        \begin{align*}
                &T^\circ_{i,j,f+1}(z) \geq T^+_{i,j,k_*}(z)\\
                \rightarrow \quad& z + 2qr^{2(f+1)/n-1} \geq z + 2(q-1)r^{2k_*/n}\\
                \rightarrow \quad& qr^{2(f+1)/n-1} \geq (q-1)r^{2k_*/n}\\
                \rightarrow \quad& \frac{q}{q-1}r^{1/n} \geq r^{2k_*/n}.
        \end{align*}
        We can further rearrange this to get
        \[q \leq \frac{r}{r-r^{2(f+1-k_*)/n}}.\]
        Taking $k_* \leq f$ demonstrates that agents $i+1,\ldots,i+f$ will all reach $x$ before agent $i+f+1$ when $q \leq \frac{r}{r-r^{2/n}}$. Similarly, taking $k_* \leq f-1$ demonstrates that agents $i+1,\ldots,i+f-1$ will all reach $x$ before agent $i+f+1$ when $q \leq \frac{r}{r-r^{4/n}}$. Now suppose that $\frac{r}{r-r^{2/n}} < q \leq \frac{r}{r-r^{4/n}}$. Then we claim that agent $i+f$ reaches its turning point $d^{(1)}_{i+f,j-1}$ somewhere in $I_{i,j}$ and agent $i+f+1$ will be the $(f+1)^{st}$ agent to reach only the points $(d_{i,j},d^{(1)}_{i+f,j-1}]$. Indeed, we can observe that
        \[d^{(1)}_{i+f,j-1} = ar^{2f/n-1}d_{i,j} = \frac{a}{r^{1/n}}d_{i,j}\]
        and since we are assuming that $a \geq r^{1/n}$ we have $d^{(1)}_{i+f,j-1} \in I_{i,j}$. We can also conclude from this that agent $i+f$ will arrive before agent $i+f+1$ to exactly those positions $x = z d_{i,j}$ for which $z \leq \frac{a}{r^{1/n}}$, i.e. $a \geq z r^{1/n}$.
\end{proof}

\begin{proof}(Lemma~\ref{lm:intersects_gen})
        To determine a general expression for $\rho_{i,j,k}$ and $\tau_{i,j,k}$ one must consider multiple cases depending on where the intersection of the trajectories of agents $i$ and $i+k$ takes place relative to their sub-turning points. Agent $i$ will either be located in the interval $[d_{i,j},d^{(1)}_{i,j}]$ or the interval $(d^{(1)}_{i,j},d_{i,j+1}]$ and agent $i+k$ will either be in $[d_{i+k,j-1},d^{(1)}_{i+k,j-1}]$ or $(d^{(1)}_{i+k,j-1},d_{i+k,j}]$. We will only be interested in the two cases that agent $i$ is in $[d_{i,j},d^{(1)}_{i,j}]$ and agent $i+1$ is in $[d_{i+k,j-1},d^{(1)}_{i+k,j-1}]$ or $(d^{(1)}_{i+k,j-1},d_{i+k,j}]$. One can refer to Figure~\ref{fig:intersections} for an illustration of each of these two cases.

        \begin{figure}
                \centering
                \includegraphics[width=\linewidth,keepaspectratio]{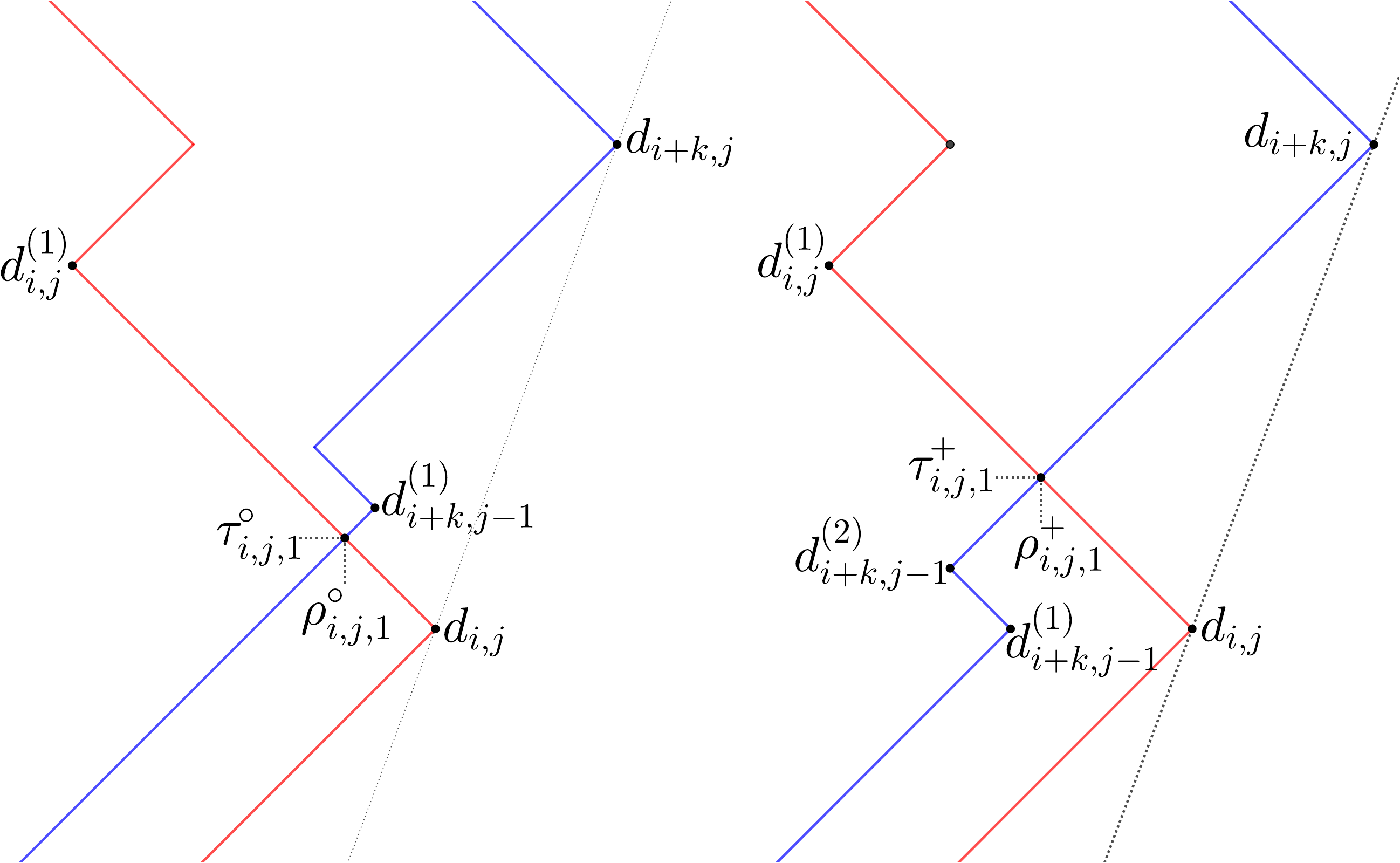}
                \caption{Illustration of the intersection points $(\rho^\circ_{i,j,k},\tau^\circ_{i,j,k})$ (left side) and $(\rho^+_{i,j,k},\tau^+_{i,j,k})$ (right side) of agents $i$ and $i+k$. In this example $j$ is even.\label{fig:intersections}}
        \end{figure}

        Let $\rho^\circ_{i,j,k}$ and $\tau^\circ_{i,j,k}$ correspond to the intersection points when agent $i$ is moving between $[d_{i,j},d^{(1)}_{i,j}]$ and agent $i+1$ is moving between $[d_{i+k,j-1},d^{(1)}_{i+k,j-1}]$ (this situation is depicted on the left side of Figure~\ref{fig:intersections}). Let $\rho^+_{i,j,k}$ and $\tau^+_{i,j,k}$ correspond to the intersection points when agent $i$ is moving between $[d_{i,j},d^{(1)}_{i,j}]$ and agent $i+1$ is moving between $(d^{(1)}_{i+k,j-1},d_{i+k,j}]$ (this situation is depicted on the right side of Figure~\ref{fig:intersections}).

        If agent $i$ is moving between $[d_{i,j},d^{(1)}_{i,j}]$ then it will be traveling in the direction $-(-1)^j$ along the line with equation
        \begin{equation}\label{eq:line_1}
                t = t_{i,j} - (-1)^j(x-d_{i,j}) = t_{i,j} + |d_{i,j}| - (-1)^j x.
        \end{equation}

        Consider first $\rho^\circ_{i,j,k}$ and $\tau^\circ_{i,j,k}$. In this case agent $i+k$ will be traveling in the direction $(-1)^j$ along the line with equation
        \begin{align*}
                t &= t_{i+k,j-1} + (-1)^j(x-d_{i+k,j-1})\\
                &= t_{i+k,j-1} + |d_{i+k,j-1}| + (-1)^j x\\
                &= r^{2k/n-1}(t_{i,j} + |d_{i,j}|) + (-1)^j x.
        \end{align*}
        Subtracting this equation from \eqref{eq:line_1} and solving for $x$ yields
        \begin{align*}                
                \rho^\circ_{i,j,k} &= \frac{(-1)^j(1-r^{2k/n-1})}{2}(t_{i,j} + |d_{i,j}|)\\
                &= \frac{(-1)^j(1-r^{2k/n-1})}{2}((2q-1)|d_{i,j}| + |d_{i,j}|)\\
                &= q(1-r^{2k/n-1})d_{i,j}.
        \end{align*}             
        Substituting this into \eqref{eq:line_1} then yields
        \begin{align*}
                \tau^\circ_{i,j,k} &= t_{i,j} + |d_{i,j}| - (-1)^j \rho_{i,j,k}\\
                &= 2q|d_{i,j}| - q(1-r^{2k/n-1})|d_{i,j}|\\
                &= q(1+r^{2k/n-1})|d_{i,j}|
        \end{align*}
        This case will occur provided that $\tau^\circ_{i,j,k} \leq t^{(1)}_{i,j}$ and $\tau^\circ_{i,j,k} \leq t^{(1)}_{i+k,j-1}$. We get from the first inequality
        \[\tau^\circ_{i,j,k} \leq t^{(1)}_{i,j}\]
        \begin{align*}
                \rightarrow \quad q(1+r^{2k/n-1})|d_{i,j}| &\leq t_{i,j} + (1+a)|d_{i,j}|\\
                &= (2q+a)|d_{i,j}|
        \end{align*}
        and finally
        \[a \geq q(r^{2k/n-1}-1).\]                
        The right hand side of this inequality is negative for all $k \leq f$ and the inequality is therefore satisfied for any $a\geq r^{1/n}$. From the inequality $\tau^\circ_{i,j,k} \leq t^{(1)}_{i+k,j-1}$ we get
        \begin{align*}
                q(1+r^{2k/n-1})|d_{i,j}| &\leq (2q+a)|d_{i+k,j-1}|\\
                &= (2q+a)r^{2k/n-1}|d_{i,j}|.
        \end{align*}
        This can be manipulated to yield
        \[a \geq q(r^{1-2k/n}-1)\]
        which is of course the condition of \eqref{eq:rho_gen}.

        Now consider $\rho^+_{i,j,k}$ and $\tau^+_{i,j,k}$. If we refer to Figure~\ref{fig:intersections} one can observe that the point $(\rho^+_{i,j,k},\tau^+_{i,j,k})$ is shifted to the left (i.e. in the direction $-(-1)^j$) and up by an amount $|d^{(2)}_{i+k,j-1}-d^{(1)}_{i+k,j-1}|$ as compared to the point $(\rho^0_{i,j,k},\tau^0_{i,j,k})$. Since $|d^{(2)}_{i+k,j-1}-d^{(1)}_{i+k,j-1}| = s|d_{i+k,j-1}| = sr^{2k/n-1}|d_{i,j}|$ we immediately find that
        \begin{align*}
                \rho^+_{i,j,k} &= \rho^\circ_{i,j,k} - (-1)^j|d^{(1)}_{i+k,j-1}-d^{(1)}_{i+k,j-1}|\\
                &= q(1-r^{2k/n-1})d_{i,j} - sr^{2k/n-1}d_{i,j}\\
                &= [q-(q+s)r^{2k/n-1}]d_{i,j}\\
                &= [q-(q-1)r^{2k/n}]d_{i,j}
        \end{align*}
        and, similarly,
        \begin{align*}
                \tau^+_{i,j,k} &= [q+(q-1)r^{2k/n}]|d_{i,j}|.
        \end{align*}

\end{proof}   

\begin{proof}(Lemma~\ref{lm:intersects_cond})
        We will only be considering values of $a$ satisfying $q(r^{1-4/n}-1) \leq a \leq q(r^{1-2/n}-1)$. Referring to Lemma~\ref{lm:intersects_gen} we can see that this implies that the trajectories of agents $i$ and $i+1$ will intersect at the point $(\rho_{i,j,1},\tau_{i,j,1}) = (\rho^+_{i,j,1},\tau^+_{i,j,1})$. On the other hand we will have $(\rho_{i,j,k},\tau_{i,j,k}) = (\rho^\circ_{i,j,k},\tau^\circ_{i,j,k})$ for all $k=2,3,\ldots,f$. We observe that $\beta^\circ_1 > \beta^\circ_2 > \ldots > \beta^\circ_f$ implying that, for odd/even $j$, agent $i$ is further to the left/right than all agents $i+2,i+3,\ldots,i+f$ during the time interval $[\tau_{i-1,j,2},\tau_{i,j,2}] = [\tau^\circ_{i-1,j,2},\tau^\circ_{i,j,2}]$. Of course, we want to show that agent $i$ is further to the left/right (for odd/even $j$) than all other agents during the time interval $[\tau_{i-1,j,1},\tau_{i,j,1}] = [\tau^+_{i-1,j,1},\tau^+_{i,j,1}]$. Since agent $i+2$ will be further to the left/right of agent $i$ after the time $\tau^\circ_{i,j,2}$ we need to either ensure that $\tau^\circ_{i,j,2} \geq \tau^+_{i,j,1}$, or, in the case that $\tau^\circ_{i,j,2} < \tau^+_{i,j,1}$, agent $i+2$ must turn around at exactly the time $\tau^\circ_{i,j,2}$. In other words, the intersection point must coincide with a turning point. This latter condition will occur precisely when $\tau_{i,j,2}$ changes from $\tau^\circ_{i,j,2}$ to $\tau^+_{i,j,2}$, i.e. when $a = q(r^{1-4/n}-1)$.

        For $a > q(r^{1-4/n}-1)$ we will need $\tau^\circ_{i,j,2} \geq \tau^+_{i,j,1}$. We have
        \[\tau^\circ_{i,j,2} \geq \tau^+_{i,j,1}\]
        \begin{align*}
                \rightarrow \quad& q(1+r^{4/n-1}) \geq q+(q-1)r^{2/n}\\
                \rightarrow \quad& q \leq \frac{r}{r-r^{2/n}}.
        \end{align*}
        This completes the proof.             
\end{proof}

\begin{proof}(Lemma~\ref{lm:s_choice})
        Agents $i+k-1$ and $i+k$ will both be located on the cone $\cone_1^+$ on the opposite side of the origin from $d_{i,j}$ at the time $\tau^+_{i+k-1,j-1,1}$. We thus need to solve the equation $\tau^+_{i+k-1,j-1,1} = T^\circ_{i,j,f+1}(1)$ for $q$. We have
        \[\tau^+_{i+k-1,j-1,1} = T^\circ_{i,j,f+1}(1)\]
        \begin{align*}
                \rightarrow \quad& [q(1+r^{2/n})-r^{2/n}] r^{2(k-1)/n-1} = 1 + 2qr^{1/n}\\
                \rightarrow \quad& q[(1+r^{2/n})r^{2(k-1)/n-1}-2r^{1/n}] = r^{2k/n-1}+1
        \end{align*}
        and finally 
        \begin{align*}
                q &= \frac{r^{2k/n-1}+1}{(1+r^{2/n})r^{2(k-1)/n-1}-2r^{1/n}}\\
                &= \frac{r^{2k/n}+r}{(1+r^{-2/n})r^{2k/n}-2r^{1+1/n}} = \hat{q}(r,k).
        \end{align*}
\end{proof}

\begin{proof}(Lemma~\ref{lm:abound})
        We need to consider the two cases corresponding to the definition of $\hat{a}(r,q)$. First consider the case that $q \leq \frac{r}{r-r^{2/n}}$. In this case we need to show that $q(r^{1-2/n}-1) > r^{1/n}$. Since $q \geq \frac{r}{r-1}$ we will demonstrate the stronger condition that 
        \[\frac{r(r^{1-2/n}-1)}{r-1} > r^{1/n} \ \  \rightarrow \ \  r^2 > (r-1)r^{3/n} + r^{1+2/n}\]
        which is clearly satisfied for $r>1$ and $n \geq 5$. Now consider the case that $q > \frac{r}{r-r^{2/n}}$. Then we need to show that $q(r^{1-4/n}-1) > r^{1/n}$. Since $q > \frac{r}{r-r^{2/n}}$ we get the stronger condition
        \[\frac{r(r^{1-4/n}-1)}{r-r^{2/n}} > 1  \quad \rightarrow \quad r^2 > r+r^{1-2/n}(r^{6/n}-1)\]
        or
        \[r > 1+r^{4/n}-\frac{1}{r^{2/n}}.\]
        For $n \geq 5$ both sides of this inequality grow with $r>1$, however, the left hand side grows faster. Moreover, when $r=1$ the two sides are equal. We can thus conclude that $\hat{a}(r,q) > r^{1/n}$.
\end{proof}

\begin{proof}(Theorem~\ref{thm:f_crash})
        Suppose without loss of generality that the target is at location $x = z d_{i,j} \in I_{i,j}$. In the worst case the first $f$ agents that reach the target are faulty and $x$ is just beyond a turning point. For the generalized schedules we will need to consider two two separate turning points. 
        
        Consider pairs $(r,u) \in \parset$ and suppose that $q = \hat{q}(r,u)$ and $a = \hat{a}(r,q)$. By definition of $\parset$ we have $r>1$ and $\frac{r}{r-1} \leq q \leq \min\{\frac{r}{r-r^{1-2/n}},\ \frac{r}{r-r^{4/n}}\}$. Then, by Lemma~\ref{lm:abound}, we also have $a > r^{1/n}$. By Lemma~\ref{lm:S_cond} we know that agent $i+f+1$ will thus reach location $x$ at the time
        \begin{align*}
                T_{i,j,f+1}(z) &= T^\circ_{i,j,f+1}(z) = (z + qr^{2(f+1)/n-1})|d_{i,j}|\\
                &= (z + qr^{1/n})|d_{i,j}|.
        \end{align*}
        Agents $i+1,\ldots,i+f-1$ will all reach $x$ before agent $i+f+1$, and agents $i+f+2,\ldots,i+n$ will reach $x$ after agent $i+f+1$. Agent $i+f$ will reach $x$ before agent $i+f+1$ provided that $q \leq \frac{r}{r-r^{2/n}}$ or $q > \frac{r}{r-r^{2/n}}$ and $a \geq zr^{1/n}$. We first consider the case that agent $i+f$ reaches $x$ before agent $i+f+1$.

        With $q \leq \frac{r}{r-r^{2/n}}$ or $q > \frac{r}{r-r^{2/n}}$ and $a \geq zr^{1/n}$ we know that agent $i+f+1$ will be the $(f+1)^{st}$ agent to reach the target. In the worst case the target is at position $x = (1+\eps)d_{i,j}$, i.e. $z = 1+\eps$. The search time is therefore
        \[S_f^x = T^\circ_{i,j,f+1}(1+\eps) = (1 + \eps + 2qr^{1/n})|d_{i,j}|.\]
        Since $q = \hat{q}(r,u)$, we know by Lemma~\ref{lm:s_choice} that agent $i+u$ will be (one of) the agents most distant from $x$ at time $S_f^x$. Moreover, this agent will be located on the cone $\cone_1^+$ at time $T^\circ_{i,j,f+1}(1)$. We can thus conclude that
        \begin{align*}
                \Delta_x^f = |x| + \frac{T^\circ_{i,j,f+1}(1)}{\beta^+_1} + \eps |d_{i,j}|.
        \end{align*}
        and the evacuation time is therefore
        \begin{align*}
                E_f^x &= S_f^x+\Delta_f^x\\
                &= T^\circ_{i,j,f+1}(1+\eps) + |x| + \frac{T^\circ_{i,j,f+1}(1)}{\beta^+_1} + \eps |d_{i,j}|\\
                &= (1+3\eps)|d_{i,j}| + \left(1+\frac{1}{\beta^+_1}\right)T^\circ_{i,j,f+1}(1)\\
                &= (1+3\eps)|d_{i,j}| + \left(1+\frac{1}{\beta^+_1}\right)(1 + 2qr^{1/n})|d_{i,j}|.
        \end{align*}
        Dividing by $|x| = (1+\eps)|d_{i,j}|$ and taking the limit $\eps \rightarrow 0$ yields the competitive ratio to be
        \begin{align*}
                R^A_f = 1 + \left(1+\frac{1}{\beta^+_1}\right)(1 + 2qr^{1/n})
        \end{align*}   
        where we have included a superscript $A$ to indicate that this is the competitive ratio of Scenario A. We have from equation \eqref{eq:betas} that
        \[\beta_1^+ = \frac{q+(q-1)r^{2/n}}{|q-(q-1)r^{2/n}|}.\]
        We claim that $q-(q-1)r^{2/n} \geq 0$ as a result of the requirement that $q \leq \frac{r}{r-r^{1-2/n}}$. Indeed, we have
        \[q-(q-1)r^{2k/n} = q(1-r^{2/n})+r^{2/n} \geq 0\]
        or
        \[q \leq \frac{r^{2/n}}{r^{2/n}-1} = \frac{r}{r-r^{1-2/n}}.\]
        Now observe that
        \begin{align*}
                1+\frac{1}{\beta^+_1} &= 1 + \frac{q-(q-1)r^{2/n}}{q+(q-1)r^{2/n}} = \frac{2q}{q+(q-1)r^{2/n}}
        \end{align*}
        and we can conclude that the competitive ratio of Scenario A is
        \begin{align*}
                R^A_f = 1 + \frac{2q(1 + 2qr^{1/n})}{q + (q-1)r^{2/n}}
        \end{align*}
        as required.

        Now consider the case that $q > \frac{r}{r-r^{2/n}}$ and $a < zr^{1/n}$. Recall that the condition $a < zr^{1/n}$ derived from the requirement that agent $i+f$ reaches its turning point $d^{(1)}_{i+f,j-1}$ before reaching $x$. The time at which agent $i+f$ reaches $x$ is $T^+_{i,j,f}(z)$ and at this time agents $i+1,\ldots,i+f-1$ and $i+f+1$ have already visited $x$. The worst case for this scenario places the target just beyond the turning point $d^{(1)}_{i+f,j-1}$ and so we take $x = (1+\eps)d^{(1)}_{i+f,j-1} = (1+\eps)\frac{a}{r^{1/n}}|d_{i,j}|$. Thus, with $z = (1+\eps)\frac{a}{r^{1/n}}$, the search time for this case is at most $T^+_{i,j,f}(z)$, i.e.
        \begin{align*}
                S_f^x &\leq T^+_{i,j,f}(z) = (z+2(q-1)r^{2f/n})|d_{i,j}|\\
                &= (z+2(q-1)r^{1-1/n})|d_{i,j}|.
        \end{align*}
        At the time $T^\circ_{i,j,f+1}(1)$ the agent most distant from $x$ was at distance $|x|+T^\circ_{i,j,f+1}(1)/\beta_1^+$. Thus, at time $T^+_{i,j,f}(z)$ this agent will be further away from $x$ by at most the distance $T^+_{i,j,f}(z)-T^\circ_{i,j,f+1}(1)$. We therefore have that
        \[\Delta_f^x \leq |x|+\frac{T^\circ_{i,j,f+1}(1)}{\beta_1^+} + T^+_{i,j,f}(z)-T^\circ_{i,j,f+1}(1).\]
        For the evacuation time we find
        \begin{align*}
                &E_f^x = S_f^x+\Delta_f^x\\
                &\leq T^+_{i,j,f}(z) + |x|+\frac{T^\circ_{i,j,f+1}(1)}{\beta_1^+} + T^+_{i,j,f}(z)-T^\circ_{i,j,f+1}(1)\\
                &= |x|-\left(1-\frac{1}{\beta_1^+}\right)T^\circ_{i,j,f+1}(1) + 2T^+_{i,j,f}(z)\\
                &= z|d_{i,j}|-\left(1-\frac{1}{\beta_1^+}\right)(1+2qr^{1/n})|d_{i,j}|\\
                &\qquad\qquad+ 2(z+2(q-1)r^{1-1/n})|d_{i,j}|\\
                &= \Biggl[3z + 4(q-1)r^{1-1/n}\\
                &\qquad\qquad- \left(1-\frac{1}{\beta_1^+}\right)(1 + 2qr^{1/n})\Biggr]|d_{i,j}|.
        \end{align*}     
        We have
        \begin{align*}
                1-\frac{1}{\beta^+_1} &= 1 - \frac{q-(q-1)r^{2/n}}{q+(q-1)r^{2/n}} = \frac{2(q-1)r^{2/n}}{q + (q-1)r^{2/n}}.
        \end{align*}
        and thus
        \begin{align*}
                E_f^x &\leq \Biggl[3z + 4(q-1)r^{1-1/n}\\
                &\qquad\qquad - \frac{2(q-1)r^{2/n}(1 + 2qr^{1/n})}{q + (q-1)r^{2/n}}\Biggr]|d_{i,j}|\\
                &= \left[3z + \frac{2(q-1)}{r^{1/n}}\left(2r - \frac{r^{3/n}(1 + 2qr^{1/n})}{q + (q-1)r^{2/n}}\right)\right]|d_{i,j}|.
        \end{align*}    
        Dividing through by $z|d_{i,j}|$ with $z = (1+\eps)\frac{a}{r^{1/n}}$ and taking the limit $\eps \rightarrow 0$ yields the competitive ratio
        \begin{align*}
                R^B_f &\leq 3 + \frac{2(q-1)}{a}\left(2r - \frac{r^{3/n}(1 + 2qr^{1/n})}{q + (q-1)r^{2/n}}\right).
        \end{align*} 
        Since we are assuming that $q > \frac{r}{r-r^{2/n}}$ we have $a = \hat{a}(r,q) = q(r^{1-4/n}-1)$ and we can finally conclude that
        \begin{align*}
                R^B_f &\leq 3 + \frac{2(q-1)}{q(r^{1-4/n}-1)}\left(2r - \frac{r^{3/n}(1 + 2qr^{1/n})}{q + (q-1)r^{2/n}}\right)
        \end{align*}      
        as required.
\end{proof}

\begin{figure}[tbhp!]
        \centering
        \includegraphics[width=\linewidth,keepaspectratio]{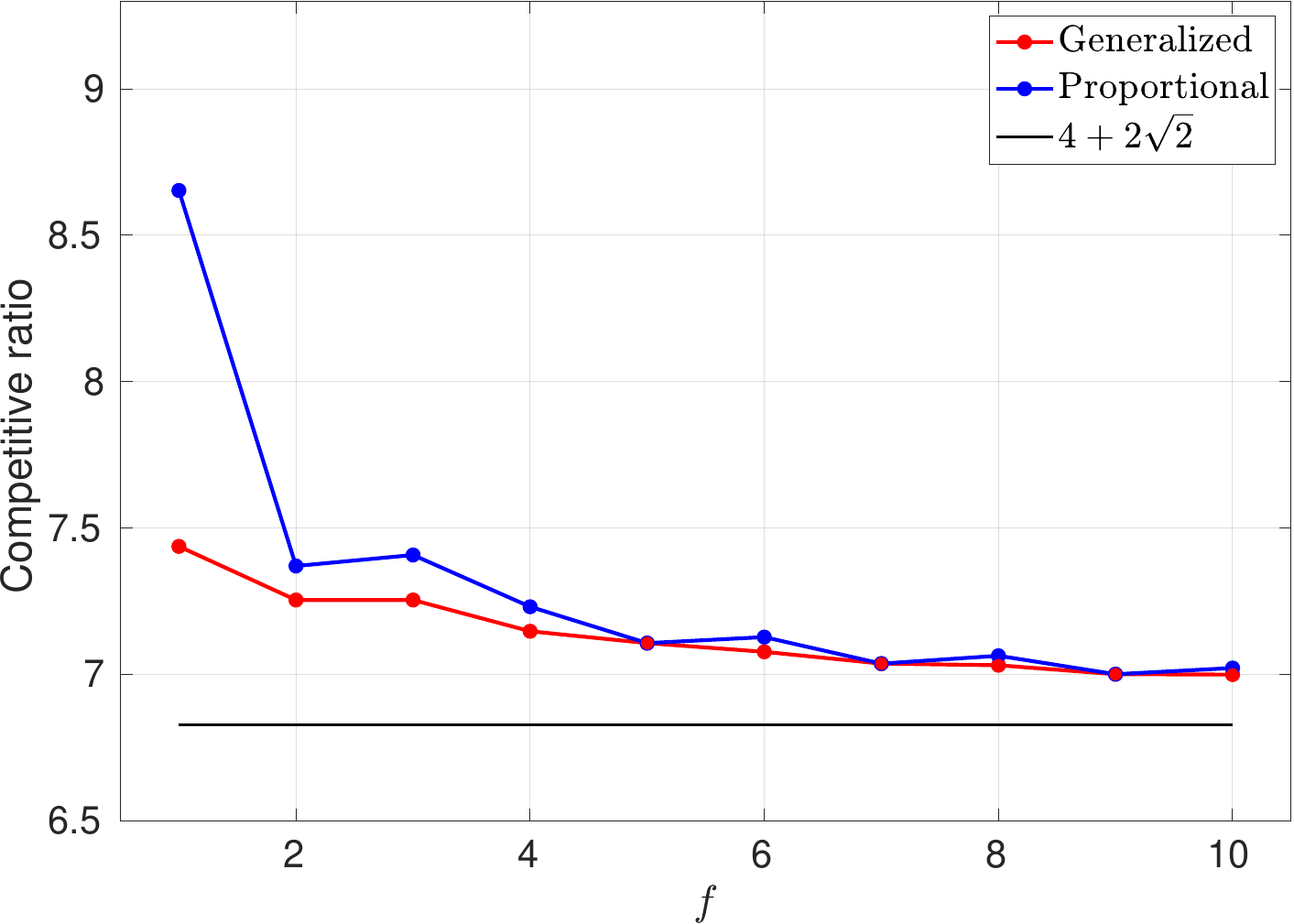}
        \caption{The optimized competitive ratio of the generalized schedules as a function of the number of faults.\label{fig:gen}}
\end{figure}            
\begin{figure}[tbhp]
        \centering
        \includegraphics[scale=0.08]{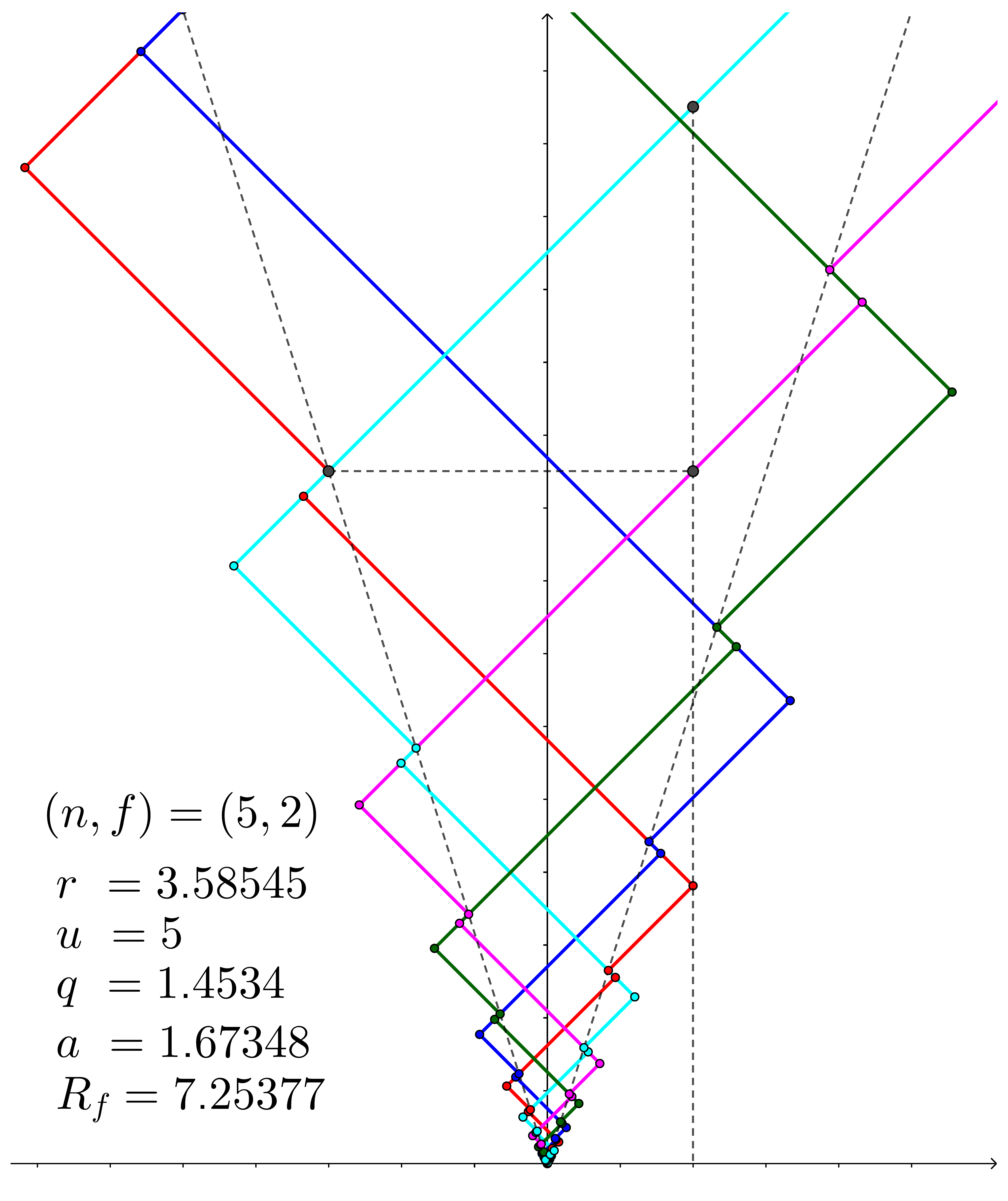}
        \includegraphics[scale=0.08]{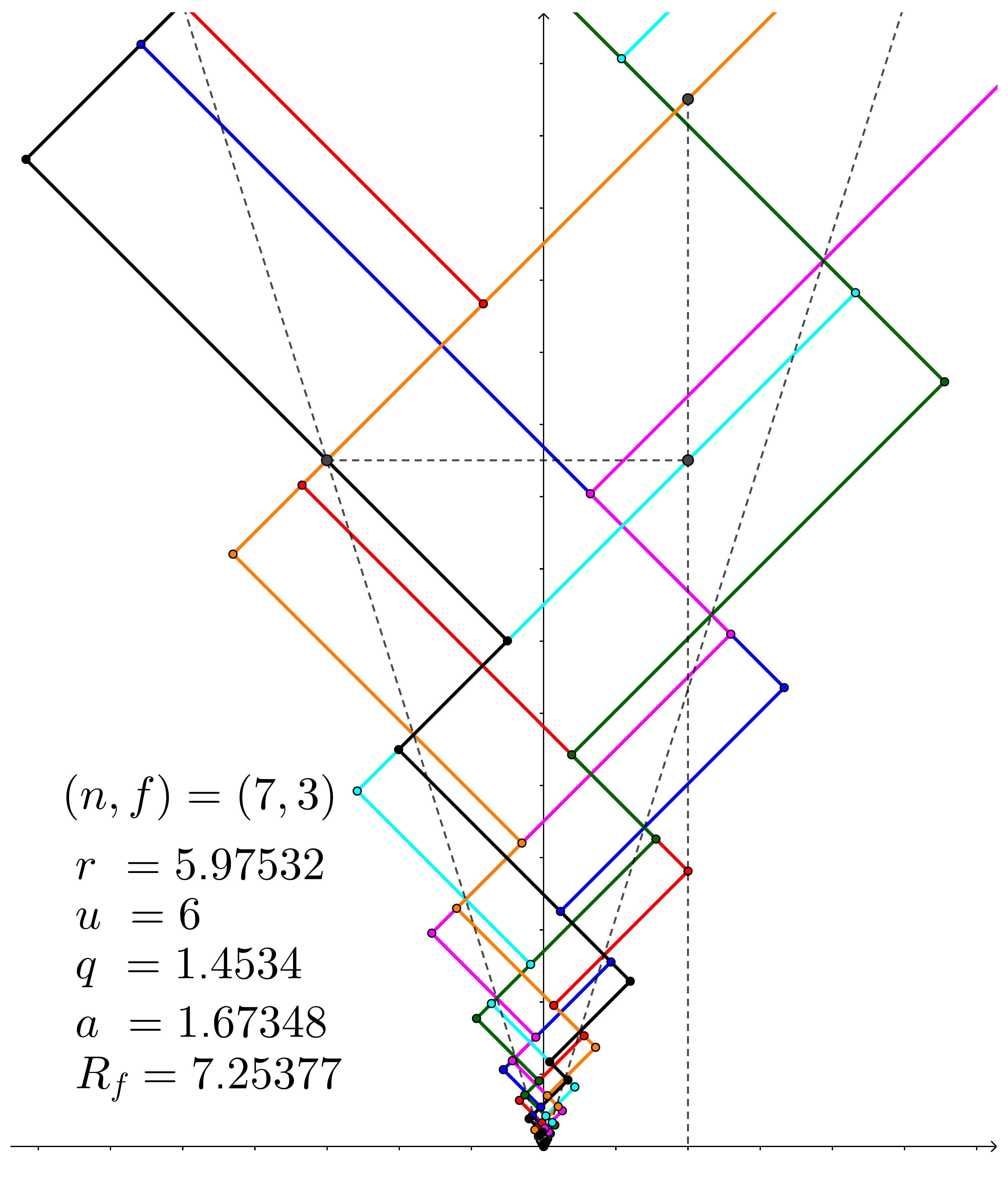}
        \includegraphics[scale=0.08]{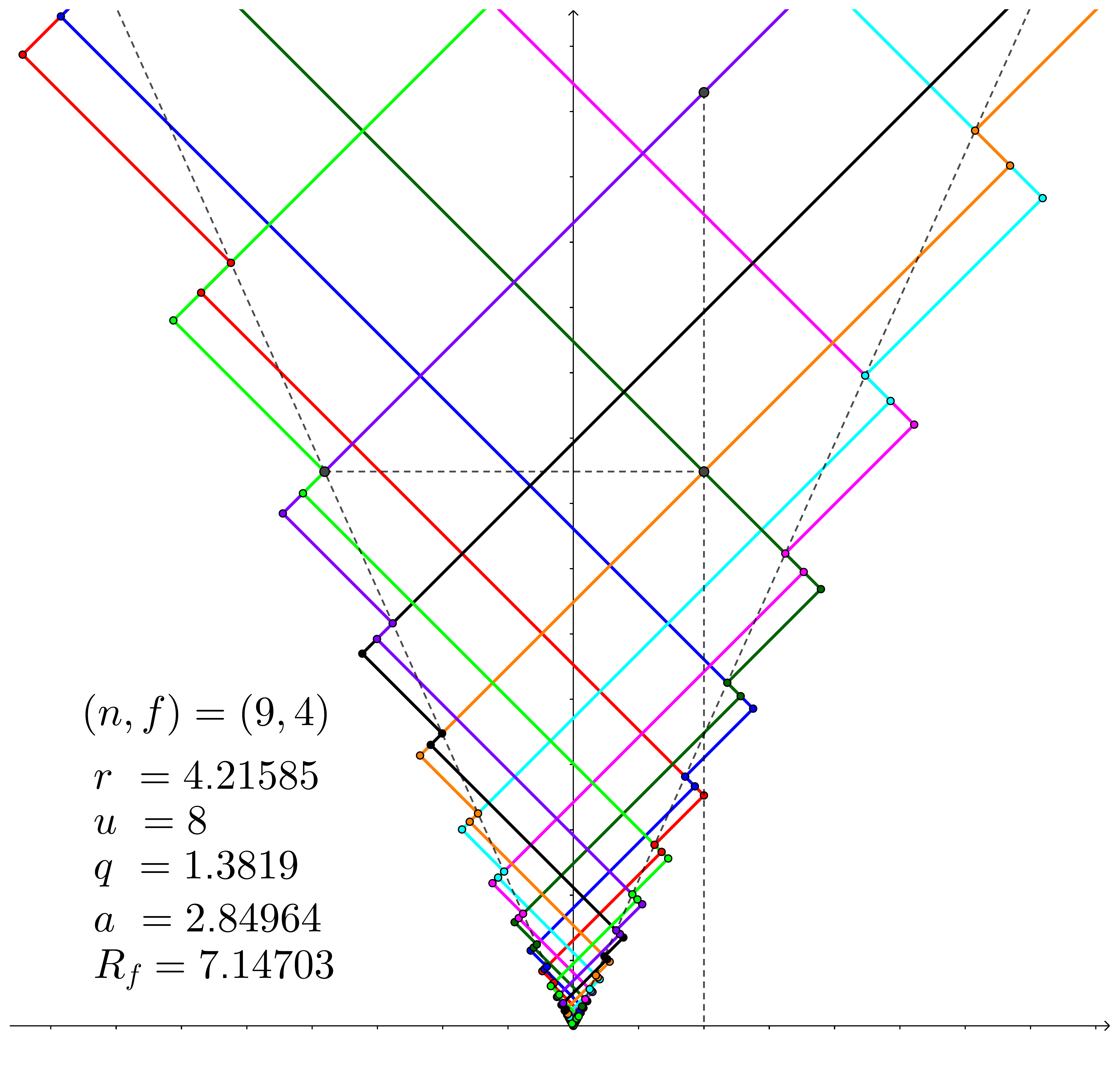}        
        \caption{The optimized trajectories of the agents when $f=2,3,4$.}\label{fig:traj}
\end{figure}
\begin{table}[tbhp]
        \centering
        \caption{Optimal parameter choices for the generalized schedules and the corresponding competitive ratios. The competitive ratio of the optimized proportional schedules is included in the last row for comparison.\label{tbl:gen}}
        \begin{tabular}{c|c|c|c|c}
                $(n,f)$& (5,2) & (7,3) & (9,4) & (11,5)\\ \hline\hline
                $r$& 3.58545 & 5.97532 & 4.21585 & 3.22306\\\hline
                $u$& 5 & 6 & 8 & 10\\\hline
                $q$& 1.45340 & 1.45340 & 1.38190 & 1.44983\\\hline
                $s$& 0.17225 & 1.25582 & 0.22813 & 0\\\hline
                $a$& 1.67348 & 1.67348 & 2.84964 & 2.32740\\\hline
                $R_f$& 7.25377 & 7.25377 & 7.14703 & 7.10648\\\hline
                $R_f$, {\footnotesize P.S.}& 7.37001 & 7.40756 & 7.23077 & 7.10648
        \end{tabular}                                
\end{table}

\subsection{Proofs missing from Subsection~\ref{sec:one_fault}}

\begin{proof}(Lemma~\ref{lm:opt_s})
        Agent $i$ will reach $d^{(2)}_{i,j}$ at time $t^{(2)}_{i,j}$ and agent $i+2$ will reach location $d_{i,j}$ at the time $t^{(2)}_{i,j}(1)$. We will thus need to solve the equation $t^{(2)}_{i,j} = T_{i,j,2}(1)$ for the parameter $q$. Since $a \geq \frac{1}{r^{1/3}}$ Lemma~\ref{lm:tijk_gen} states that agent $i+2$ will reach $d_{i,j}$ at the time $T_{i,j,2}(1) = T^{\circ}_{i,j,2}(1)$. We therefore need to solve
        \begin{align*}
                t^{(2)}_{i,j} &= T_{i,j,2}(1)\\
                \rightarrow \quad& q(r+1)-r+a = 1 + 2q r^{2(2)/n-1}\\
                \rightarrow \quad& q(r+1-2r^{1/3}) = r+1-a
        \end{align*}
        and finally
        \[q = \frac{r+1-a}{r+1-2r^{1/3}}\]
        as required.       
\end{proof}

\begin{proof}(Lemma~\ref{lm:RA})
        With $q$ given by \eqref{eq:opt_s} we know that agent $i+2$ will reach $d_{i,j}$ at the time $T_{i,j,2}(1) = t^{(2)}_{i,j}$ and at this time agent $i$ will be located at $d^{(2)}_{i,j}$. In the worst case the target is just beyond $d_{i,j}$ and the competitive ratio is therefore
        \begin{equation*}
                R^A_1 = 1 + \frac{t^{(2)}_{i,j} - |d^{(2)}_{i,j}|}{|d_{i,j}|} = 1 + \frac{t^{(1)}_{i,j} + |d^{(1)}_{i,j}|}{|d_{i,j}|}.
        \end{equation*}
        With $t^{(1)}_{i,j} = (2q+a)|d_{i,j}|$ and $|d^{(1)}_{i,j}| = a|d_{i,j}|$ we get
        \begin{align*}
                R^A_1 = 1+2(q+a).
        \end{align*}
\end{proof}

\begin{proof}(Lemma~\ref{lm:RB})
        In the worst case the target is at location $x$ just beyond $d^{(1)}_{i+1,j-1}$. Referring to Figure~\ref{fig:31_crash} one can observe that the announcement in this case will not change the trajectory of agent $i+2$ and the evacuation will complete at the time agent $i+2$ would normally reach $x$. The last turning point agent $i+2$ visits before visiting $x$ is the turning point $d_{i+2,j-1}$ and thus the competitive ratio is
        \begin{align*}
                R^B_1 = 1 + \frac{t_{i+2,j-1}+|d_{i+2,j-1}|}{d^{(1)}_{i+1,j-1}}.
        \end{align*}
        With $t_{i+2,j-1} = (2q-1)r^{1/3}|d_{i,j}|$, $|d_{i+2,j-1}| = r^{1/3} |d_{i,j}|$,  and $|d^{(1)}_{i+1,j-1}| = ar^{-1/3} |d_{i,j}|$ we get
        \begin{align*}
                R^B_1 = 1 + \frac{(2q-1)r^{1/3} + r^{1/3}}{ar^{-1/3}} = 1 + \frac{2qr^{2/3}}{a}.
        \end{align*}     
\end{proof}

\begin{proof}(Lemma~\ref{lm:opt_a1})
        We need to solve the equation $R^A_1 = R^B_1$ for $a$. We have
        \[1+2(q+a) = 1 + \frac{2qr^{2/3}}{a}\]
        or
        \[q(a-r^{2/3})+a^2 = 0.\]    
        Substituting in the expression \eqref{eq:opt_s} gives
        \begin{align*}
                &\left(\frac{r+1-a}{r+1-2r^{1/3}}\right)(a-r^{2/3}) + a^2 = 0\\
                \rightarrow \quad& (r+1-a)(a-r^{2/3}) + a^2(r+1-2r^{1/3}) = 0
        \end{align*}
        and after a little more manipulation we arrive at
        \begin{align*}
                a^2(r-2r^{1/3}) + a(r+1+r^{2/3}) - r^{2/3}(r+1) = 0.
        \end{align*}
        Let $P(a)$ represent the polynomial in $a$ on the left of this equation. There are two possible solutions to $P(a) = 0$:
        \begin{multline*}
                a = \frac{1}{2(r-2r^{1/3})}\Biggl[-(r+1+r^{2/3})\\ \pm \sqrt{(r+1+r^{2/3})^2+4r^{2/3}(r+1)(r-2r^{1/3})}\Biggr].
        \end{multline*}
        Note that with our assumption that $r > 2\sqrt{2}$ the denominator of $a$ will be positive. Clearly, then, only the positive root can result in $a>0$. To show that this value of $a$ indeed lies within the range $[\frac{1}{r^{1/3}},r^{1/3}]$ we make use of the polynomial $P(a)$. When $r > 2\sqrt{2}$ the coefficient of $a^2$ is positive and so the parabola curves upward. Moreover, when $a = r^{1/3}$ we have
        \begin{align*}
                P(r^{1/3}) &= r^{2/3}(r-2r^{1/3}) + r^{1/3}(r+r^{2/3}+1)\\
                &\qquad\qquad - r^{2/3}(r+1)\\
                &= r^{5/3}-2r + r^{4/3}+r+r^{1/3}-r^{5/3}-r^{2/3}\\
                &= -r + r^{4/3}+r^{1/3}-r^{2/3}\\
                &= r(r^{1/3}-1) - r^{1/3}(r^{1/3}-1)\\
                &= (r-r^{1/3})(r^{1/3}-1)
        \end{align*}
        which is clearly positive. On the other hand, when $a = 1/r^{1/3}$ we have
        \begin{align*}
                &P(r^{-1/3}) = r^{-2/3}(r-2r^{1/3}) + r^{-1/3}(r+r^{2/3}+1)\\
                &\qquad\qquad - r^{2/3}(r+1)\\
                &= r^{1/3}-2r^{-1/3} + r^{2/3} + r^{1/3} + r^{-1/3} - r^{5/3} - r^{2/3}\\
                &= 2r^{1/3}-r^{-1/3} - r^{5/3}
        \end{align*}
        and it is simple to confirm that this is negative for $r > 2\sqrt{2}$.
\end{proof}

\section{Proof of Theorem~\ref{thm:31_search}}
\begin{proof}(Theorem~\ref{thm:31_search})
        We base our search algorithm off of the generalized schedule of Theorem~\ref{thm:31_crash}. In particular, we will only consider the generalized trajectory corresponding to the optimized choices of the parameters $(r,q,a)$, i.e. $r = 6.833921$, $a = 1.699557$, and $q = 1.518949$. We will represent by $\alpha$ the competitive ratio for evacuation corresponding to these parameter choices, i.e. $\alpha = 7.43701137$.
                       
        At the beginning of our search algorithm we let the agents follow their generalized schedule trajectories until the moment an announcement is made. Suppose that this announcement claims that the target is at location $x \in I_{i,j}$. Then the behavior of the agents will depend on whether or not $x \in (d_{i,j},d^{(1)}_{i+2,j-1}]$ or $x \in (d^{(1)}_{i+2,j-1},d_{i+1,j}]$. We consider first the case that $x \in (d_{i,j},d^{(1)}_{i+2,j-1}]$.

        \paragraph*{Case 1, $\mathbf{x \in (d_{i,j},d^{(1)}_{i+2,j-1}]}$:} Let us assume without loss of generality that $x>0$. An illustration of this case is provided in Figure~\ref{fig:search1} for reference. When the first announcement claims that the target is at location $x \in (d_{i,j},d^{(1)}_{i+2,j-1}]$ all agents that have not visited $x$ at the time of the announcement will immediately move to $x$ to check the claim. One can observe that the agents will arrive to $x$ in the order $i+1$, $i+2$, $i$. If agent $i$ is the one who made the announcement then agents $i+1$ and $i+2$, having already visited $x$, will immediately know that agent $i$ is faulty. They will thus immediately proceed to move in opposite directions at full speed until the target is found. This will clearly not be a worst case and so we will suppose that it is one of agents $i+1$ or $i+2$ that made the announcement. 
        
        \begin{figure}
                \centering
                \includegraphics[scale=0.65]{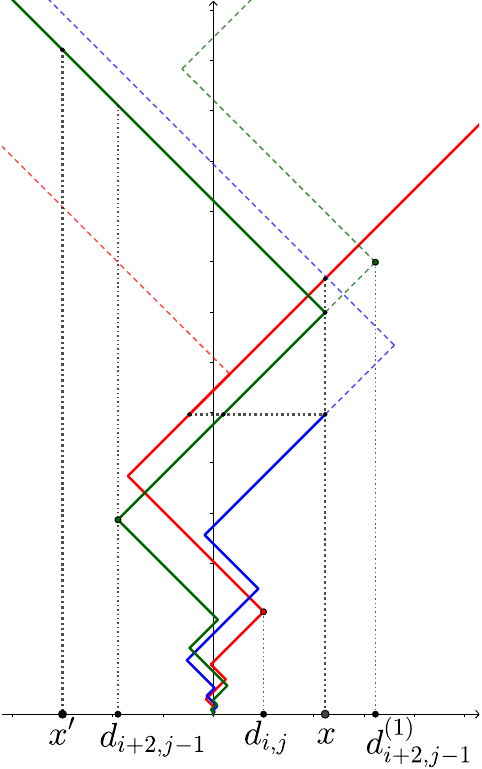}\hspace{2cm}
                \includegraphics[scale=0.65]{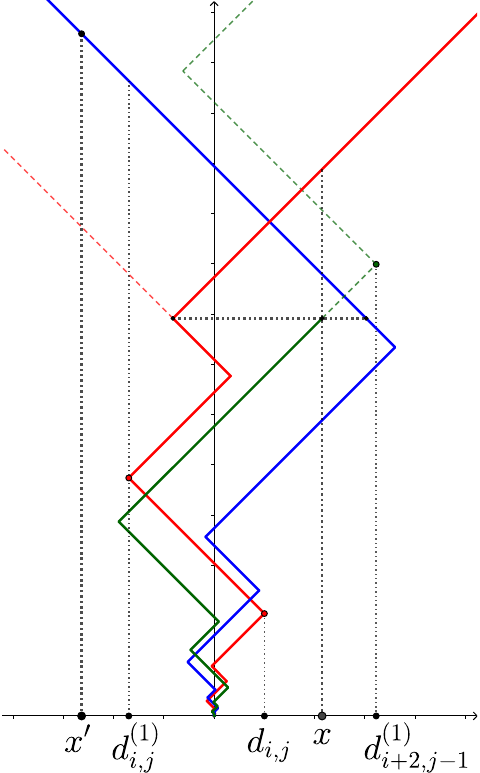}
                \caption{Setup for the case that the first announcement claims that the target is at location $x \in (d_{i,j},d^{(1)}_{i+2,j-1}]$ and the actual target is on the left side at location $x'$. Agent $i$ is in red, agent $i+1$ in blue, and agent $i+2$ in green. On the left side it is agent $i+1$ that falsely announces the target and on the right side it is agent $i+2$ that falsely announces the target. In either case the dashed lines show where the agents would be if the announcement were not made and the solid lines show the actual trajectories of the agents.\label{fig:search1}}
        \end{figure}

        In this case the first reliable agent (either $i+1$ or $i+2$) that reaches $x$ will immediately move to the left after reaching $x$ and continue in this direction until the real target is found. Agent $i$, who will be moving to the right when it reaches $x$, will continue moving to the right after $x$ until either it, or the other reliable agent, finds the real target. It is clear from the results on evacuation that agent $i$ will be able to find a target at any position $x'$ to the right of $x$ in time less than $\alpha x'$. Thus, we can assume that the target is on the left side of the origin. 
        
        Assume first that it is agent $i+1$ that made the erroneous announcement. This situation is depicted on the left side of Figure~\ref{fig:search1}. Since agent $i+2$ will reach $x$ after its turning point $d_{i+2,j-1}$ and since $x \leq d^{(1)}_{i+2,j-1}$ the exact time agent $i+2$ will reach $x$ is $t_{i+2,j-1}+|d_{i+2,j-1}|+x$. Moreover, the real target must be somewhere to the left of $d_{i+2,j-1}$, since, otherwise, agent $i+2$ would have announced it already. Thus, the competitive ratio will be 
        \begin{align*}
                R = 1+\sup \frac{t_{i+2,j-1}+|d_{i+2,j-1}|+2x}{|x'|}
        \end{align*}
        where the supremum is over $x \in (d_{i,j},d^{(1)}_{i+2,j-1}]$ and $x' < d_{i+2,j-1}$. Since the ratio increases with $x$ and decreases with $x'$ we have
        \begin{align*}
                R &= 1+\frac{t_{i+2,j-1}+|d_{i+2,j-1}|+2|d^{(1)}_{i+2,j-1}|}{|d_{i+2,j-1}|}\\
                &= 1+\frac{t^{(1)}_{i+2,j-1}+|d^{(1)}_{i+2,j-1}|}{|d_{i+2,j-1}|}
        \end{align*}
        This ratio is independent of $i$ and $j$ and so we also have
        \begin{align*}
                R = 1+\frac{t^{(1)}_{i,j}+|d^{(1)}_{i,j}|}{|d_{i,j}|}.
        \end{align*}
        This, however, is exactly the expression for $R^A_1$ in Lemma~\ref{lm:RA} which equals $\alpha$. 

        Now consider the situation depicted on the right side of Figure~\ref{fig:search1} where agent $i+2$ makes the announcement. Then agent $i+1$ will have already visited $x$ when the announcement is made and will immediately proceed to move left until it finds the target at $x'$. Referring to Figure~\ref{fig:search1} one can observe that at the time of the announcement the leftmost point visited by a reliable agent will be the point $d^{(1)}_{i,j}$and we must therefore have $x' < d^{(1)}_{i,j}$. One can also observe that agent $i+1$ will reach $x'$ at the time $t_{i+1,j} + |d_{i+1,j}| + x'$ and so the competitive ratio in this situation will be
        \begin{align*}
                R = 1+\frac{t_{i+1,j} + |d_{i+1,j}|}{|d^{(1)}_{i,j}|} = 1+\frac{t_{i+2,j-1} + |d_{i+2,j-1}|}{|d^{(1)}_{i+1,j-1}|}
        \end{align*}
        and this is identical to the expression for $R^B_1$ in Lemma~\ref{lm:RB}. We can thus conclude that the competitive ratio of the search is $\alpha$ when $x \in (d_{i,j},d^{(1)}_{i+2,j-1}]$.

        \paragraph*{Case 2, $\mathbf{x \in (d^{(1)}_{i+2,j-1},d_{i+1,j}]}$:} This situation is depicted in Figure~\ref{fig:search2}. As before, we assume that $x>0$ and, since agent $i$ will be the last agent to reach $x$, we may assume that it is not agent $i$ that announces the target. Then, we will first send agent $i$ to its turning point $d_{i,j+1}$ before having it move to check the claim at $x$. The remaining reliable agent will move to the left after it visits $x$. If agent $i$ finds the target before reaching $x$ it will announce this and the search will end once the remaining reliable agent visits $x$. This will not be a worst case and so we assume that agent $i$ does not find the target before reaching $x$.

        \begin{figure}
                \centering
                \includegraphics[scale=0.3]{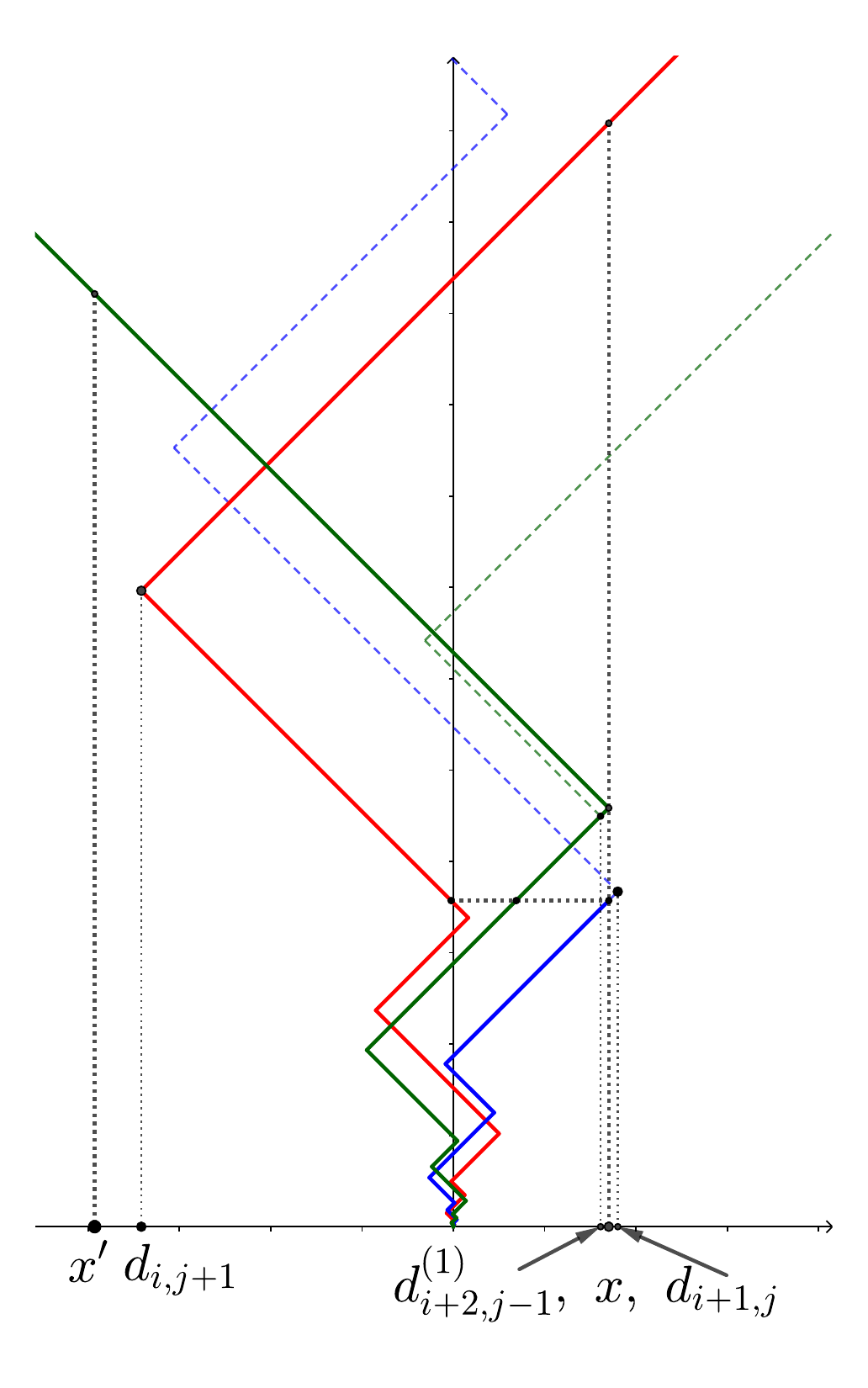}\hspace{2cm}
                \includegraphics[scale=0.3]{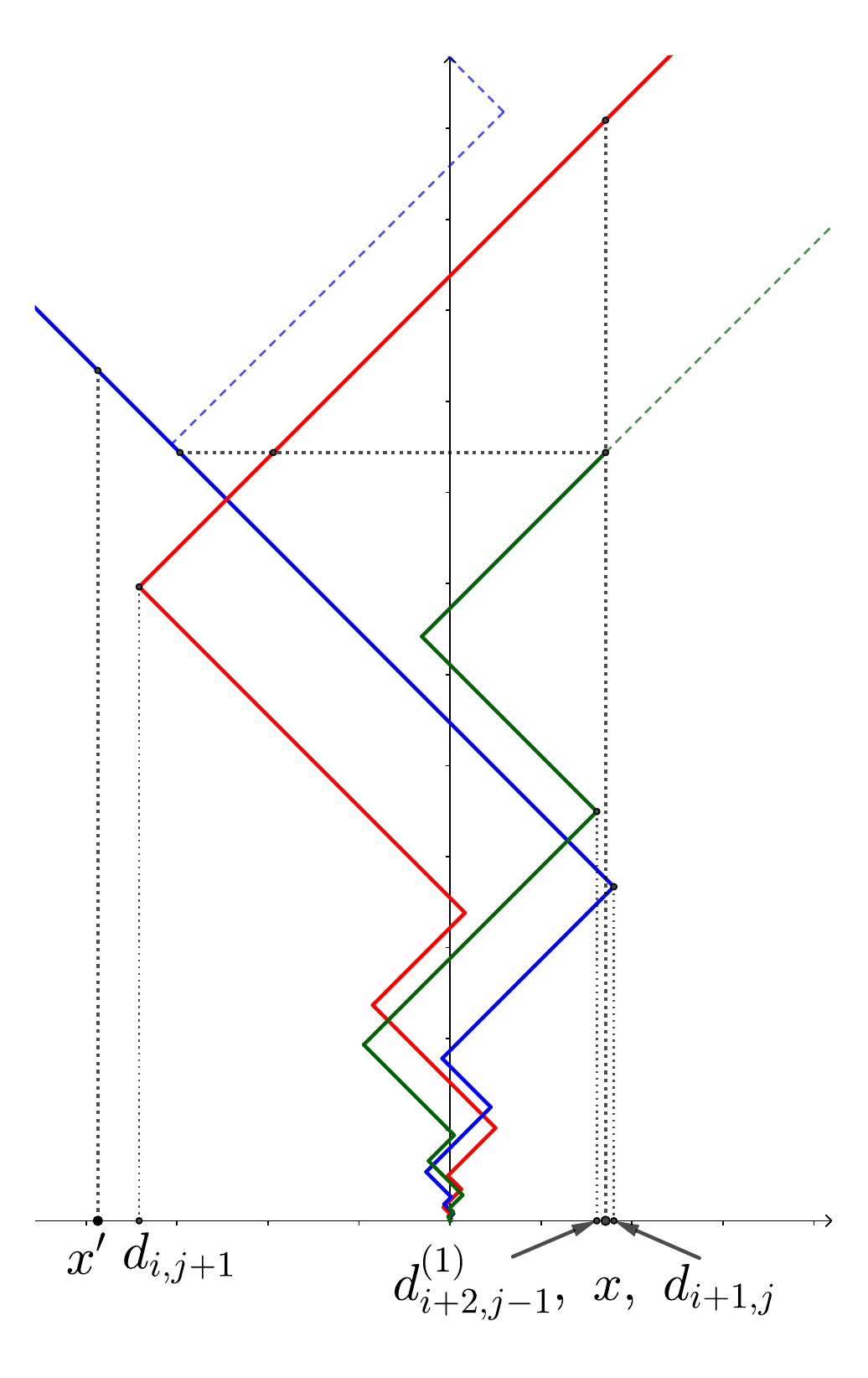}
                \caption{Setup for the case that the first announcement claims that the target is at location $x \in (d^{(1)}_{i+2,j-1},d_{i+1,j}]$ and the actual target is on the left side at location $x'$. Agent $i$ is in red, agent $i+1$ in blue, and agent $i+2$ in green. On the left side it is agent $i+1$ that falsely announces the target and on the right side it is agent $i+2$ that falsely announces the target. In either case the dashed lines show where the agents would be if the announcement were not made and the solid lines show the actual trajectories of the agents.\label{fig:search2}}
        \end{figure}

        Suppose that it is agent $i+1$ that announces the target and refer to the left side of Figure~\ref{fig:search2}. If the target is at location $x'$ to the right of $x$ then agent $i$ will reach the target after it visits its turning point $d_{i,j+1}$ and so the time at which agent $i$ reaches $x'$ is $t_{i,j+1}+|d_{i,j+1}|+x'$ and the competitive ratio will be
        \begin{align*}
                R &= 1 + \sup_{x'}  \frac{t_{i,j+1}+|d_{i,j+1}|}{x'} = 1+\frac{t_{i,j+1}+|d_{i,j+1}|}{d^{(1)}_{i+2,j-1}}\\
                &= 1+\frac{t_{i+2,j-1}+|d_{i+2,j-1}|}{d^{(1)}_{i+1,j-1}}
        \end{align*}
        where we have used our assumption that $x' > x > d^{(1)}_{i+2,j-1}$ and the fact that the ratio is independent of $i$ and $j$. This expression is, of course, the same as $R^B_1$ from Lemma~\ref{lm:RB} and so in this case the competitive ratio will be at most $\alpha$. Thus, we consider the case that the target is at location $x'$ to the left of the origin. In this case, the target must be at location $x' < d_{i,j+1}$ since otherwise agent $i$ would have already found it. Agent $i+2$ will receive the announcement as it is moving away from $d_{i+2,j-1}$ and so it will reach $x$ at the time $t_{i+2,j-1} + |d_{i+2,j-1}|+x$. The earliest it could reach the target at $x'$ will then be $t_{i+2,j-1} + |d_{i+2,j-1}|+2x+|x'|$. The competitive ratio is then
        \begin{align*}
                R = 1 + \sup \frac{t_{i+2,j-1} + |d_{i+2,j-1}|+2x}{|x'|}
        \end{align*}
        where the supremum is over $x \in (d^{(1)}_{i+2,j-1},d_{i+1,j}]$ and $x' < d_{i,j+1}$. We thus find that
        \begin{align*}
                R &= 1 + \frac{t_{i+2,j-1} + |d_{i+2,j-1}|+2|d_{i+1,j}|}{|d_{i,j+1}|}\\
                &= 1 + \frac{r^{1/3}(2q-1)|d_{i,j}| + r^{1/3}|d_{i,j}|+2r^{2/3}|d_{i,j}|}{r|d_{i,j}|}\\
                &= 1 + \frac{2(q+r^{1/3})}{r^{2/3}}.
        \end{align*}
        Substituting in the values $r = 6.833921$, and $q = 1.518949$ then yields $R = 2.897498 < \alpha$. Thus, when agent $i+1$ announces the target the competitive ratio is at most $\alpha$.

        Now suppose that it is agent $i+2$ that announces the target. Then we have the situation depicted on the right of Figure~\ref{fig:search2}. One can easily observe from this figure that this situation cannot be worse than the case when agent $i+1$ announces the target. Indeed, agent $i$ will reach a target $x' > x$ at the same time for either case, and agent $i+1$ will reach a target $x' < d_{i,j+1}$ earlier than agent $i+2$ would in the situation that $i+1$ was the announcing agent. Thus, the competitive ratio will be at most $\alpha$ when $x \in (d^{(1)}_{i+2,j-1},d_{i+1,j}]$. This completes the proof.
\end{proof}

\end{document}